\documentclass[orivec,
			 envcountsect,
			 envcountsame]{llncs}
\pdfoutput=1
\usepackage{cite}
\usepackage[T1]{fontenc}
\usepackage{mathtools}
\usepackage{macros}
\usepackage{mathtools,
  amsmath,
  amsfonts,
            stmaryrd,
            fdsymbol,
            bbm,
            caption,
            tikz-cd,
            mathrsfs,
            float}
\usepackage{xcolor}
\usepackage{multicol}
\usepackage{scalerel}
\usepackage{tikz}
\usepackage{nicematrix}
\usepackage{todonotes}
\usepackage[hidelinks]{hyperref}
\usepackage[capitalise,nameinlink]{cleveref}
\crefname{equation}{}{}			
\usetikzlibrary{automata, positioning, calc}

\newcommand{\takeout}[1]{\empty}
\numberwithin{equation}{section}		%equations numbered within section
\spnewtheorem{notation}[theorem]{Notation}{\bfseries}{}

\title{An Expressive Coalgebraic Modal Logic \\ for Cellular Automata}

\usepackage{orcidlink}
\author{Henning Basold\inst{1}\orcidlink{0000-0001-7610-8331} \and 
   	   Chase Ford\inst{2}\orcidlink{0000-0003-3892-5917} \and 
   	   Lulof Pir\'{e}e\inst{3}\orcidlink{0009-0004-4802-6514}
}

 \institute{Leiden University, Netherlands \email{h.basold@liacs.leidenuniv.nl} \\
   \and Leiden University, Netherlands \email{m.c.ford@liacs.leidenuniv.nl} \\
   \and Tallinn University of Technology, Estonia \email{lulof.piree@taltech.ee}
 }

\begin{document}
\maketitle

\begin{abstract}
Cellular automata provide models of parallel computation
based on cells, whose connectivity is given by an action of a
monoid on the cells. At each step in the computation, every
cell is decorated with a state that evolves in discrete steps
according to a local update rule, which determines the next
state of a cell based on its neighbour's states. In this paper,
we provide a coalgebraic view on cellular automata, which
does not require typical restrictions, such as uniform
neighbourhood connectivity and uniform local rules. Using
the coalgebraic view, we devise a behavioural equivalence
for cellular automata and a modal logic to reason about their
behaviour. We then prove a Hennessy-Milner style theorem,
which states that pairs of cells satisfy the same modal formulas
exactly if they are identified under cellular behavioural equivalence.

\end{abstract}

\section{Introduction}

Cellular automata provide a model for computations that are distributed in space~\cite{Bhattacharjee20,Ilachinski01:CellularAutomataDiscrete,Sarkar00:BriefHistoryCellular}.
Besides being an immediate model for parallel computing, cellular automata find applications in optimisation~\cite{BOM15:NaturalComputingAlgorithms}, reversible computing~\cite{Morita17:TheoryReversibleComputing}, quantum computing~\cite{LLW+23:PhotonicElementaryCellular}, physics~\cite{ZPT21:StepBunchesNanowires}, chemistry~\cite{Ulam62:MathematicalProblemsConnected,CW22:AsymmetricCellDEVSModels}, biology~\cite{VonNeumann66:TheorySelfreproducingAutomata,DD17:CellularAutomatonModeling}, sociology~\cite{MP07:ComplexAdaptiveSystems}, and group theory~\cite{CC10:CellularAutomataGroups}.
Cellular automata (CA) comprise a collection of cells that hold a state and interact with other cells in their neighbourhood in order to update their state.
In this paper, we will focus on synchronous cellular automata that update cell states in discrete rounds, and where the parallelism only stems from spatial distribution.
A cellular automaton starts from an initial \emph{configuration}, that is, an assignment of states to each cell, and evolves by means of local rules that compute the next state of the cell from the states of its neighbours.
The spatial neighbourhood of a cell is given by a collection of elements of a monoid, which in turn acts on the cells to provide the spatial connectivity~\cite{roka}.

Let us explore as example a cellular automaton that presents a (flawed) attempt at leader election in a cyclic network~\cite{Santoro06:DesignAnalysisDistributed}.
The network has six nodes, given by the modular group $\Z_{6}$ arranged as the circle in \cref{fig:cycle-ca} on the left.
\begin{figure}[ht]
  \begin{equation*}
    \begin{tikzpicture}[x=0.8cm, y=0.8cm]
      \draw(0,0) circle (1);
      \pgfmathsetmacro\n{6}
      \foreach \i in {0,1,...,5} {
        \pgfmathsetmacro\r{\i*(360/\n)}
        \fill (\r:1) circle (1pt) coordinate (n-\i);
        \node (n-\i) at (\r:1.2) {\i};
      };
      \draw[thick,red,->] ([shift=(5:1.4)]0,0) arc (3:35:1.4) node[midway, right]{$+1$};
      \draw[thick,red,->] ([shift=(-5:1.4)]0,0) arc (-3:-35:1.4) node[midway, right]{$-1$};
    \end{tikzpicture}
    \quad
    \begin{aligned}[b]
      % & M = \Z \\
      & N = \set{-1, 0, 1}
      & & \gamma_{1}(k)(n) = (k + n) \bmod 6, \quad (k\in\Z_{6}, n\in\Z) \\
      & S = \N
      & & \gamma_{2}(k)(f) = (f(-1) + f(0) + f(1)) \bmod (k + 1) \\
      & \phantom{x}
      & & \qquad \qquad \text{for } k \in \Z_{6} \text{ and } f \from N \to S
    \end{aligned}
  \end{equation*}
  \vspace*{-5ex}
  \caption{Cellular automaton, where the spatial connectivity is given by the natural action $\gamma_{1}$ of $\Z$ on $\Z_{6}$ and local rules $\gamma_{2}$ for cells $k \in \Z_{6}$ by modular arithmetic.}
  \label{fig:cycle-ca}
\end{figure}
The spatial connectivity of this cellular automaton is given by the natural action of the additive monoid $\Z$ on $\Z_{6}$, which we denote by $\gamma_{1}$.
Thus an integer $n \in \Z$ corresponds to taking $n$ steps anti-clockwise if $n$ is positive or clockwise if it is negative, see \cref{fig:cycle-ca}.
Every cell can access the state of itself and its direct neighbours, which is modelled by the neighbourhood $N \subset \Z$ with $N = \set{-1, 0, 1}$.
This attempt at leader election updates in every round the state of a cell $k$ by adding the states of a cell and its neighbours modulo $k+1$.
A leader is elected once cells no longer change their state and the leader will be the cell with the highest number.
The global behaviour of the cellular automaton is given by the evolution of configurations. \Cref{tab:cyclic-ca-behaviour} shows two sequences of behaviour, where the left ends with an elected leader, while the right shows periodic behaviour without the election of a leader.
\begin{table}[ht]
  \centering
  \begin{NiceTabular}{w{c}{1.2cm}|cccccc}
    \RowStyle[bold]{}
    \rule[-2mm]{0pt}{0.6cm} % Force height of row and keep 4mm space below in other columns
    \diagbox{step}{cell} & 0 & 1 & 2 & 3 & 4 & 5 \\ \hline
    0 & 0 & 0 & 2 & 0 & 0 & 0 \\
    1 & 0 & 0 & 2 & 2 & 0 & 0 \\
    2 & 0 & 0 & 1 & 0 & 2 & 0 \\
    3 & 0 & 1 & 1 & 3 & 2 & 2 \\
    4 & 0 & 0 & 2 & 2 & 2 & 4 \\
    5 & 0 & 0 & 1 & 2 & 3 & 0 \\
    \RowStyle[rowcolor=red!50, nb-rows=2]{}   %rounded-corners
    6 & 0 & 1 & 0 & 2 & 0 & 3 \\
    7 & 0 & 1 & 0 & 2 & 0 & 3
  \end{NiceTabular}
  \qquad \qquad
  \begin{NiceTabular}{w{c}{1.2cm}|cccccc}
    \RowStyle[bold]{}
    \rule[-2mm]{0pt}{0.6cm} % Force height of row and keep 4mm space below in other columns
    \diagbox{step}{cell} & 0 & 1 & 2 & 3 & 4 & 5 \\ \hline
    \RowStyle[rowcolor=red!50]{}	%rounded-corners
    0 & 0 & 0 & 1 & 0 & 1 & 4 \\
    1 & 0 & 1 & 1 & 2 & 0 & 5 \\
    2 & 0 & 0 & 1 & 3 & 2 & 5 \\
    3 & 0 & 1 & 1 & 2 & 0 & 1 \\
    4 & 0 & 0 & 1 & 3 & 3 & 1 \\
    5 & 0 & 1 & 1 & 3 & 2 & 4 \\
    6 & 0 & 0 & 2 & 2 & 4 & 0 \\
    \RowStyle[rowcolor=red!50]{}	%rounded-corners
    7 & 0 & 0 & 1 & 0 & 1 & 4
  \end{NiceTabular}
  \vspace*{1mm}
  \caption{Evolution of different configurations for the cellular automaton in \cref{fig:cycle-ca}. Left: a fixed point is reached in line 6; right: configuration with period 7.}
  \label{tab:cyclic-ca-behaviour}
\end{table}

There are various ways of implementing leader election~\cite{HS80:DecentralizedExtremafindingCircular}, which can be modelled with cellular automata.
However, correctness proofs for distributed algorithms can be tedious.
In this paper, we propose a \emph{modal logic to reason about the behaviour of cellular automata}, in which properties such as periodicity and nilpotency (existence of fixed points) can be expressed.
Our main result is a Hennessy-Milner style theorem, showing that the logical equivalence induced by our logic captures precisely behavioural equivalence of cellular automata.

In order to formulate behavioural equivalence of cellular automata, we also contribute a \emph{coalgebraic description of cellular automata} in the class of, what is traditionally referred to as, synchronous, non-regular and non-uniform cellular automata~\cite{Bhattacharjee20}.
Understanding simulations and bisimulations for such models is generally challenging~\cite{roka}, but our coalgebraic perspective provides a clear understanding of homomorphism and bisimulations between cellular automata via an extension of the standard notion of coalgebra homomorphisms~\cite{Rutten00} and coalgebraic Aczel-Mendler (or span) bisimulations~\cite{staton_bisim} to, what we call, \emph{cellular morphisms} and \emph{cellular bisimulations}.
We show that cellular morphisms and bisimulations are precisely what is required to preserve and reflect the global behaviour of cellular automata, and that these notions of behavioural equivalence correspond to logical equivalence in our logic.

\paragraph*{Related work}
We refer to Bhattacharjee~et~al.~\cite{Bhattacharjee20} for a comprehensive account and comparison of the various classes of cellular automata and, in particular, for an overview of notions of non-uniformity in models of cellular networks. 
A more historical account is given in~Sarkar~\cite{Sarkar00:BriefHistoryCellular}.
Nishio~\cite{Nishio06,Nishio07} provides an analysis of the relative expressive power between uniform cellular automata and those with non-uniform neighbourhood structure.
Many notions of behavioural equivalence and simulation orders for specific cellular automata classes have appeared in the literature, typically involving a fixed monoid and binary state set~\cite{DR99,handbook_ch6}.
Our notion of behavioural equivalence generalises these notions, although we do not cover simulations.

Modal logics for cellular automata with similar modalities and atomic formulas as ours
have been described by Delivorias~et~al.~\cite{ca_tdl} and
Hagiya~et~al.~\cite{ca_via_templog}, but their semantics are given only for uniform CA on regular grids of the form $\Z^{d}$
and they lack an analogue of our Hennessy-Milner style theorem. % \Cref{theo:hm_theorem}.
Hagiya~et~al. provide an explicit algorithm for satisfiability checking, while we focus on logical foundations and not on algorithms or complexity.

Other authors have proposed several logics with orthogonal goals in mind. 
Ishida~et~al.~\cite{proplog_for_local_rule,formulae_on_monoids} investigate the use of propositional logic as a framework for specifying the local rules of binary cellular automata; Zhang and Bölcskei~\cite{many_val_logic} provide an extension thereof using propositional {\L}ukasiewicz logic.
Das and Chakraborty~\cite{das_formal_logic} present a deduction system capable of computing the global rule of a fixed uniform cellular automaton with finite neighbourhood structure.

There have been a number of approaches to modelling cellular automata network structure via category-theoretic machinery. 
Capobianco and Uustalu~\cite{unif_comon_cas,grad_comon_cas} model synchronous uniform cellular automata as morphisms in the Kleisli category of the cowriter comonad on the category of uniform spaces and uniformly continuous maps.
This gives a static description of CA but no immediate way to induce the dynamic behaviour of CA which is needed to interpret a modal logic.
In contrast, we package the cellular space and local rules into coalgebras, rather than coalgebra morphisms, which leads  directly to a model of cellular behaviour. 
Our model enables us to capture transitions between configurations of cells in the manner sketched above.

Widemann and Hauhs~\cite{Widemann11:DistributiveLawCellularAutomata} provide a model of two-dimensional uniform cellular automata, but no logic.
Their approach is based on giving a syntax for moves along the grid $\Z^{2}$ and then equipping coalgebras for terms over this syntax with bialgebraic semantics~\cite{Klin-GSOS,Turi97:TowardsMathOS} by constructing a distributive law.
In contrast, we provide a general coalgebraic description of cellular automata which enables us to handle arbitrary synchronous CA and, more importantly, allows us to extract a general definition of cellular behavioural equivalence.

Staton~\cite{staton_bisim} describes and relates various notions of bisimulation native to coalgebras for a given type functor.
We make essential use of, what he calls, Aczel-Mendler bisimulation for relating cells in cellular automata locally.
However, in order to deal with the global behaviour on configurations, we enhance this definition with an additional relation on configurations subject to suitable compatibility conditions.
This means that the standard coalgebraic definitions are not directly applicable and our notion of cellular bisimulation is non-trivial.
That being said, our functor to model cellular automata as coalgebras preserves weak pullbacks, which implies that all four definitions of coalgebraic bisimulation are equivalent~\cite{staton_bisim}.
We leave it for future work to extend this equivalence to cellular bisimulations.

Basic modal logic~\cite{BRV01} has been extended to the level of coalgebras in the seminal work of Moss~\cite{Moss99}.
Subsequently, Pattinson~\cite{pattinson_sem_principles,pattinson_pred_lift} introduced predicate liftings
as a framework for establishing expressiveness results linking logical equivalence and behavioural equivalence on the level of coalgebras~\cite{klin,schroeder_pred_lift}. For an overview, see Kupke and Pattinson~\cite{kupke_coalg_mollog_survey}  and C\^{i}rstea~et al.~\cite{cirstea_coalg_mollog_survey}. 
We note that this framework does not directly yield our expressiveness result (\Cref{theo:hm_theorem}) as our notion of cellular morphism refines the usual coalgebra homomorphisms. It would be interesting to understand whether our description of cellular automata could be understood in the context of open maps~\cite{joyal_1996_open_maps}.

\paragraph*{Outline}
The remainder of this paper is organised as follows.
%\Cref{sec:related_work} gives a brief overview of related literature.
After the preliminaries in \cref{sec:prelim}, we present in \cref{sec:ca} our coalgebraic definition of cellular automata.
The resulting notion of coalgebra morphism is refined in \cref{sec:sim} to cellular morphisms: these preserve global behaviour and can be considered as the right notion of behavioural equivalence.
In~\cref{sec:relation-equivalence}, we introduce our notion of cellular bisimulation along with a construction of the largest such relation.
\Cref{sec:logic} introduces the syntax and semantics of our modal logic, and gives examples of its expressivity.
In \cref{sec:hennessy_milner}, we prove the correspondence between behavioural and logical equivalence.
We conclude with a summary and suggestions for future work in \cref{sec:conclusion}.
%
%%%%%%%%%%%%%%%%%%%%%%%%%%%%%%%%%%%%%%%%%%%%%%%%%
%%%%%%%%%%%%%%%%%%%%%%%%%%%%%%%%%%%%%%%%%%%%%%%%%
%%%%%%%%%%%%%%%%%%%%%%%%%%%%%%%%%%%%%%%%%%%%%%%%%
%%%%%%%%%%%%%%%%%%%%%%%%%%%%%%%%%%%%%%%%%%%%%%%%%
%%%%%%%%%%%%%%%%%%%%%%%%%%%%%%%%%%%%%%%%%%%%%%%%%
%%%%%%%%%%%%%%%%%%%%%%%%%%%%%%%%%%%%%%%%%%%%%%%%%
%%%%%%%%%%%%%%%%%%%%%%%%%%%%%%%%%%%%%%%%%%%%%%%%%
%Section:Preliminaries
%
\section{Preliminaries}\label{sec:prelim}
We assume familiarity with basic concepts
from category theory~\cite{AHS90,MacLane98}.
We briefly gather the necessary background on
(universal) coalgebra~\cite{Rutten00} and
comonads~\cite{BW90}.

Unless mentioned otherwise, we work
over the category~$\Set$ of sets and maps which 
is complete and cocomplete. We fix a terminal 
object $1 = \{\star\}$, and we write~$!_X$ for the 
unique map~$X\to 1$. Furthermore,~$\Set$ admits
the structure of a \emph{Cartesian closed category}%
~\cite[Ch.~VII]{AHS90}; for each pair of sets~$Y, Z$,
the \emph{internal hom}~$[Y,Z]$ is just the set of all 
maps from $Y$ to $Z$. Cartesian closure means that 
there is natural bijection~$\Set(X\times Y, Z)\cong \Set(X, [Y, Z])$. 
From right to left, it maps $f\colon X\to Z^Y$ to its \emph{transpose} $\lTransp{f} \colon X\times Y\to Z$ given by $f(x, y) = f(x)(y)$.
% \begin{equation}
%   \label{eqn:curry}
%   \lTransp{f} \colon X\times Y\to Z
%   \quad \text{with} \quad
%   f(x, y) = f(x)(y)
% \end{equation}
%

%%%%%%%%%%%%%%%%%%%%%%%%%%%%%%%%%%%%% 
\subsection{Coalgebra}
(Universal) coalgebra~\cite{Rutten00,Jacobs16,Adamek05} is a
framework for the uniform analysis of (state-based)
transition structures. We first give the general definition
before providing a specialised discussion of coalgebras
on~$\Set$.

\begin{definition}\label{defn:coalgebra}
Let $F\colon\BC\to\BC$ be a functor on a
category~$\BC$. An~\emph{($F$-)coalgebra} is a
pair~$(X, \gamma)$ consisting of an object~$X\in\BC$ and
a morphism~$\gamma\colon X\to FX$.
A \emph{homomorphism} from~$(X, \gamma)$ to $(Y, \xi)$ 
is a morphism $h\colon X\to Y$ in $\BC$ such that 
$Fh \comp \gamma = \xi \comp h$.
% the following diagram commutes:
% \begin{equation}\label{eqn:homomorphism}
% \begin{tikzcd}
%  X \ar[d,  "\chi" '] \ar[r, "h"]	& Y \ar[d, "\xi"]  \\
%  FX \ar[r, "Fh"] 					& FY
% \end{tikzcd}
% \end{equation}
%
We write $\Coalg(F)$ for the category of coalgebras 
and their homomorphisms, and we denote by 
$\U{-} \from \Coalg(F) \to \BC$ the faithful forgetful 
functor that maps a coalgebra to its underlying carrier.
\end{definition}

We often omit the carriers of coalgebras from the notation and simply write~$\gamma$ in lieu of~$(X,\gamma)$.
Thus,~$|\gamma| = X$.
Slightly abusing notation, we write $h$ instead of $|h|$ for the underlying map of a coalgebra homomorphism $h\colon\gamma\to\delta$.

Throughout the paper, we will follow the established cellular
automata terminology: \emph{cells} are elements of the carrier of
coalgebras, while \emph{states} are labels assigned to those cells
by coalgebras. The intuition is that cells generate the spatial
layout of cellular automata, while the labelling represents the
current state of a cell that accumulates into the global computational
state of a cellular automaton. This deviates from the common
terminology in coalgebra, where the elements of the carrier
are typically referred to as states, thereby borrowing intuition
from automata theory. However, this terminology does not align
well with that of cellular automata because the state of a system
arises only from the labelling. For instance, in the game of life,
cells provide a 2-dimensional grid layout, which has no computational
content without a labelling of ``dead'' and ``alive'' cells~\cite{DR99}.

\begin{example}\label{expl:exponent}
Let $M$ be a set.
Using the Cartesian closure, we obtain a functor $\intHom{M}{-} \from \Set \to \Set$
that acts by post-composition, that is
\begin{equation*}
\intHom{M}{f}(g) = M\xra{g} X\xra{f} Y.
\end{equation*}
We denote the functor $\intHom{M}{-}$ by $\cowriterfunc$ and call it the \emph{co-writer} functor, see \cref{example:cowriter_comonad}.
An $\cowriterfunc$-coalgebra is a set $X$ equipped with a map $\gamma\colon X \to \intHom{M}{X}$.
Writing $x \xra{m} y$ for $x, y \in X$ and $m \in M$ with $\gamma(x)(m) = y$, we view $\gamma$ as a \emph{deterministic $M$-labelled transition system}.
A map $h\colon X\to Y$ is a homomorphism~$(X, \gamma)\to(Y,\xi)$ precisely if $h \comp \gamma(x) = \xi(h(x))$ for all $x\in X$.
In terms of the transition function, this means that $h(x) \xra{m} h(y)$ in $\xi$ if $x \xra{m} y$ in $\gamma$.
\end{example}

\subsection{Comonads}
A \emph{comonad} on $\BC$ is a functor~$F\colon\BC\to\BC$ equipped with
natural transformations~$\eps \from F \to \id_{\BC}$ (the
\emph{counit}) and~$\delta \from F \to FF$
(the \emph{comultiplication}) such that the following
diagrams commute.
\begin{equation}\label{dia:comonad_laws}
\begin{tikzcd}[column sep=29, row sep = 28]					%counit
        				& F \ar[dl, "\delta" ']
            	     		       \ar[dr, "\delta"]	 \\
FF  \ar[r, "F\eps"]    	& F  \ar[u, "\id"] 			& FF\ar[l, swap, "\eps F"]
\end{tikzcd}
\qquad\qquad
\begin{tikzcd}[row sep = 28]					%coassociativity
F \ar[r, "\delta"] \ar[d, "\delta" ']   & FF \ar[d, "F\delta"] \\
FF \ar[r, "\delta F"] 	   & FFF
\end{tikzcd}
\end{equation}
We refer to these diagrams as the \emph{counit law}
(left) and \emph{coassociativity law} (right).

An \emph{Eilenberg-Moore coalgebra} for a comonad~$(F, \epsilon, \delta)$ is an~$F$-coalgebra $(X,\gamma)$ satisfying the coherence laws expressed by the following commutative diagrams.
\begin{equation}\label{dia:comonad_coalg_laws}
\begin{tikzcd}
    X
    \ar[r, "\gamma"]
    \ar[dr, "\id" ']
    & FX
        \ar[d, "\eps"]
    \\
    & X
\end{tikzcd}
%%%%%%%%%%%%%%%%%
\qquad\qquad
%%%%%%%%%%%%%%%%%
\begin{tikzcd}
    X
        \ar[r, "\gamma"]
        \ar[d, swap, "\gamma"]
    & FX
        \ar[d, "F\gamma"]
    \\
    FX
        \ar[r, "\delta"]
    & FFX
\end{tikzcd}
\end{equation}
We write~$\coEM(F)$ for the full subcategory of~$\Coalg(F)$
spanned by the Eilenberg-Moore coalgebras of the comonad
$(F, \eps, \delta)$; this notation should not lead to confusion 
since the comonad structure will be fixed to be the cowriter 
comonad for most of the article.

\paragraph*{The cowriter comonad}\label{example:cowriter_comonad}
Let $(M,\monop,\monid)$ be a monoid.
Then the functor $\cowriterfunc$ from \cref{expl:exponent} carries the structure of a comonad.
Its counit $\eps_{X} \from [M, X]\to X$ is defined by $\eps(f)=f(\monid)$, while the comultiplication $\delta\colon[M, X]\to[M, [M, X]]$ is given by $\delta(f)(m)(n) = f(n\monop m)$.
We follow Ahman and Uustalu~\cite{AU17} in referring
to this as the \emph{cowriter comonad} induced by the monoid~$(M, \monop,\monid)$.
The category $\coEM(\cowriterfunc)$ is complete and cocomplete, and limits and colimits are
created by the forgetful functor~$\U{-} \from \coEM(\cowriterfunc) \to \Set$.
Indeed, colimits are created by any comonadic functor, and limit
creation follows from the fact that~$\cowriterfunc$ is also a right adjoint.
These claims follow from the dual of~\cite[Exercise IV2.2]{MacLane98}.

The Eilenberg-Moore coalgebras of the cowriter comonad take on the familiar form of monoid actions.
A \emph{left action} of $M$ on a set $X$ is a map~$\alpha\colon M\times X\to X$ such that
\begin{equation}\label{eq:left_action}
\alpha(\monid, x) = x
\qquad\text{and}\qquad
\alpha(m, \alpha(n, x)) = \alpha (m\monop n, x)
\end{equation}
for all~$x\in X$ and $m,n\in M$.
A set equipped with a left action of~$M$ is an \emph{$M$-set}.

We denote by $\MSet$ the category of~$M$-sets and \emph{equivariant maps}.
This means that a morphism $(X, \alpha)\to (Y, \beta)$ in $\MSet$ is a map $h\colon X\to Y$, such that $h(\alpha(m, x)) = \beta(m, h(x))$ for all~$m\in M$ and~$x\in X$.
There is a faithful forgetful functor $U \from \MSet \to \Set$, which maps an action to the underlying set.

\begin{lemma}
\label{lemma:cowriter_is_left_action}
The categories $\EM(\cowriterfunc)$ and $\MSet$ are isomorphic as concrete categories, i.e., the isomorphism preserves the underlying carrier sets and maps.
\end{lemma}

%takeout old version, I dislike how much space this takes
\takeout{
\begin{align}
    \label{eq:cowriter_counit}
    %\eps \colon\Set(\mon, X)\to X
    & \eps_X \colon X^\mon \to X
    \qquad\qquad
    & \eps_X (f\colon \mon\to X) \bydef f(\monid),
    \\
    \shortintertext{and comultiplication}
    \label{eq:cowriter_comult}
    %\delta_X \colon \Set(\mon, X)\to \Set(\mon, \Set(\mon, X))
    & \delta_X \colon X^\mon \to (X^\mon)^\mon
    \qquad\qquad
    & \delta(f)(m)(n) \bydef f(n\monop m).
\end{align}
}
%end takeout old version
%
%%%%%%%%%%%%%%%%%%%%%%%%%%%%%%%%%%%%%%%%%%%%%%%%%
%%%%%%%%%%%%%%%%%%%%%%%%%%%%%%%%%%%%%%%%%%%%%%%%%
%%%%%%%%%%%%%%%%%%%%%%%%%%%%%%%%%%%%%%%%%%%%%%%%%
%%%%%%%%%%%%%%%%%%%%%%%%%%%%%%%%%%%%%%%%%%%%%%%%%
%%%%%%%%%%%%%%%%%%%%%%%%%%%%%%%%%%%%%%%%%%%%%%%%%
%%%%%%%%%%%%%%%%%%%%%%%%%%%%%%%%%%%%%%%%%%%%%%%%%
%%%%%%%%%%%%%%%%%%%%%%%%%%%%%%%%%%%%%%%%%%%%%%%%%
%Section:Coalgebra
%
\section{Cellular automata}\label{sec:ca}

In this section, we introduce our coalgebraic definition of cellular automata and discuss several examples.
Cellular automata consist of cells arranged in a spatial structure, while their behaviour is determined by the evolution of cell states over time~\cite{roka}.
For example, in Conway's game of life, the cells are arranged in the discrete plane $\Z^{2}$ and each cell can be in one of the states dead or alive.
The state of a cell evolves by means of local rules that determine, given the state of each cell in the neighbourhood of the cell, its next state.
For the game of life, the neighbourhood comprises the direct neighbours (horizontal, vertical and diagonal) of a cell.
In fact, we can see the cells of the game of life as being generated from a chosen cell by taking steps up/down and left/right.
Such moves are precisely captured by the monoid $\Z^{2}$, while the neighbours are reached via moves in the subset $\setDef{(x, y) \in \Z^{2}}{x, y \in \set{1, 0, -1}}$.

More generally, cellular automata are based on a monoid $(M, \monop, \monid)$ that describes possible moves between cells, a subset $N$ of $M$ with inclusion $i \from N \hra M$ that describes the \emph{neighbourhood} of local rules, and a set $S$ of \emph{states} that individual cells can be in.
A cellular automaton then comprises a set of cells, a description of the cell connectivity via a monoid action $\alpha \from M \times X \to X$ and for every cell $x \in X$ a local rule $\gamma_{x}$ that computes the next state of $x$ when given the state of each cell in the neighbourhood $N$ of $x$.

There is a minor complication when it comes to local rules that does not appear in the game of life.
Na\"{i}vely, we may hope that local configurations are maps $f\colon N\to S$, to which $\gamma_{x}$ assigns a state~$\gamma_x(f)\in S$.
% is a map $\intHom{N}{S} \to S$ which  to each \emph{local configuration} 
%$f \in \intHom{N}{S}$ is a \emph{local configuration} of states in the neighbourhood of $x$.
However, this view is inadequate for treating cellular automata whose action is not free, that is, if there are $m, n \in M$ with $m \neq n$ and $\alpha(m, x) = \alpha(n, x)$.
For example, equipping the mod-2 group $\Z_{2}$ with its natural $\Z$-action is not free because~$0$ can be reached by adding any odd number to $1$: $(2k + 1) + 1 = 0 \mod 2$.
Thus, if we want to inspect the two neighbours via $N = \set{-1, 1}$, we need to ensure that the map $f \from N \to S$ additionally fulfils $f(-1) = f(1)$.
We call such maps \emph{local configurations}.

Following \cref{lemma:cowriter_is_left_action}, we know that a left action on a set $X$ of cells is equivalently presented by an Eilenberg-Moore coalgebra $c \from X \to \cowriterfunc X$.
Thus, for every $x \in X$, we get a map $c(x) \from M \to X$ and then $c$ is free if $c(x)$ is injective for every~$x$.
Hereafter, we will focus on maps $a \from M \to X$ at a distinguished cell~$a(\monid)\in X$.
In this notation, a local rule should be a map $\gamma$ that takes as argument a local configuration $f$ at $a(\monid)$ consistent with $a$, that is, $f$ should be a map~$N\to S$ such that $f(m) = f(n)$ whenever $a(i(m)) = a(i(n))$.
Since the image of $a$ is the orbit of $a(\monid)$, we call such an $f$ \emph{orbit-invariant}.
In what follows, we construct a set $I^{a}$ consisting of all local configurations as an equaliser.
This construction comes with a universal property that we use to easily obtain the functor involved in our coalgebraic model of cellular automata.

Define the \emph{orbits}~$E^{a}$ of~$a\colon M\to X$ along with projections~$p^{a}_{k} \from E^{a} \to N$ by
\begin{equation*}
  E^{a} = \setDef{(m ,n)\in N\times N}{a(i(m)) = a(i(n))}
  \, \text{ and } \,
  p^{a}_{k}(m_{1}, m_{2}) = m_{k}.
\end{equation*}
This is the kernel pair of $a \comp i$, where $i\colon N\hra M$ is the inclusion, arising from the pullback in the left of \cref{fig:orbit}.
\begin{figure}[ht]
  \centering
  \begin{tikzcd}[sep=large]
    E^{a} \dar[swap]{p_{2}^{a}} \rar{p_{1}^{a}}
    \ar[dr, phantom, "\pullback", very near start]
    & N \dar{a \comp i}
    \\
    N \rar[swap]{a \comp i}
    & X
  \end{tikzcd}
  \qquad \qquad
  \begin{tikzcd}[column sep=large]
    E^{a} \ar[drr, bend left, "p^{a}_{1}"] \ar[ddr, bend right, swap, "p^{a}_{2}"] \ar[dr, dashed, "r^{a}_{h}"]
    \\
    & E^{h \comp a} \rar{p^{h \comp a}_{1}} \dar{p^{h \comp a}_{2}} & N \dar{h \comp a \comp i}
    \\
    & N \rar[swap]{h \comp a \comp i} & Y
  \end{tikzcd}
  \caption{Orbit of $M \xra{a} X$ as pullback (left) and restriction along $X \xra{h} Y$ (right).}
  \label{fig:orbit}
\end{figure}

Let $h \from X \to Y$ be a map.
By the universal property of~$E^{a}$ and associativity of composition, we have that $(h \comp a \comp i) \comp p^{a}_{1} = (h \comp a \comp i) \comp p^{a}_{2}$, which yields a unique inclusion $r^{a}_{h}$
as in \cref{fig:orbit} on the right.
The set of orbit-invariant maps is then the following equaliser, where $c^{a}_{k}(f) = f \comp p^{a}_{k}$.
\begin{equation*}
  \begin{tikzcd}
    I^{a} \rar{j^{a}}
    & \intHom{N}{S}
    \ar[r, shift left, "c^{a}_{1}"{above}]
    \ar[r, shift right, "c^{a}_{2}"{below}]
    & \intHom{E^{a}}{S}
  \end{tikzcd}
\end{equation*}
Concretely,~$I^{a}$ consists of all maps~$f \from N \to S$ with $c^{a}_{1}(f) = c^{a}_{2}(f)$. 
This means that for all $n,m\in M$ with $a(m) = a(n)$, we have $(n, m) \in E^{a}$
and thus
\begin{equation*}
  f(n) = (f \comp p^{a}_{1})(n, m) = c^{a}_{1}(f)(n, m) = c^{a}_{2}(f)(n, m) = (f \comp p^{a}_{2})(n, m) = f(m) \, .
\end{equation*}

The following lemma shows that all maps are orbit-invariant, thus local configurations, if $a$ is injective on the neighbourhood.
This simplifies reasoning about so-called free cellular automata, see \cref{def:ca} below.
\begin{lemma}
\label{lem:orb-inv}
If $a\comp i$ is injective, then $E^{a}\cong N$ and $I^{a} \cong \intHom{N}{S}$.
\end{lemma}

\noindent
Next, note that the inclusion $r^{a}_{h}$ (\cref{fig:orbit}) yields an inclusion
$t^{a}_{h} \from I^{h \comp a} \to I^{a}$ using
\begin{equation*}
  (c^{a}_{1} \comp j^{h \comp a})(f)
  = f \comp p^{a}_{1}
  = f \comp p^{h \comp a}_{1} \comp r^{a}_{h}
  = f \comp p^{h \comp a}_{2} \comp r^{a}_{h}
  = f \comp p^{a}_{2}
  = (c^{a}_{2} \comp j^{h \comp a})(f)
\end{equation*}
and the universal property of the equaliser as in the following diagram:
\begin{equation*}
  \begin{tikzcd}[column sep=large]
    I^{a} \rar{j^{a}}
    & \intHom{N}{S}
    \ar[r, shift left, "c^{a}_{1}"{above}]
    \ar[r, shift right, "c^{a}_{2}"{below}]
    & \intHom{E^{a}}{S}
    \\
    I^{h \comp a} \rar{j^{h \comp a}} \ar[u, dashed, "t^{a}_{h}"]
    & \intHom{N}{S}
    \ar[r, shift left, "c^{h \comp a}_{1}"{above}]
    \ar[r, shift right, "c^{h \comp a}_{2}"{below}]
    \uar{\id}
    & \intHom{E^{h \comp a}}{S}
    \uar[swap]{\intHom{r^{a}_{h}}{S}}
  \end{tikzcd}
\end{equation*}
This allows us to define a functor~$C \from \Set \to \Set$ by
\begin{equation}\label{eq:functor}
  CX = \coprod_{a \from M \to X} \intHom{I^{a}}{S}
  \quad \text{and} \quad
  (Ch)(a, f) = (h \comp a, f \comp t^{a}_{h})
\end{equation}
Indeed, preservation of identities and composition follows from
the unique mapping properties of $t$ and $r$.
We write~$(a, f)$ to indicate that $f\in CX$ sits in the~$a^{th}$ coproduct summand.
Moreover, there is a natural transformation $\pi \from C \to \cowriterfunc$, given by $\pi_{X}(a, f) = a$, which induces a functor $\overline{\pi} \from \Coalg(C) \to \Coalg(\cowriterfunc)$ by post-composition~\cite[Proposition A.3]{HermidaJacobs97:IndCoindFib}.

\begin{definition}
\label{def:ca}
The category~$\CApre$ of cellular automata and \emph{pre-cellular morphisms} is the full subcategory of~$\Coalg(C)$ with a functor $\tilde{\pi} \from \CApre \to \EM(\cowriterfunc)$ that makes the following diagram a pullback of categories.
  \begin{equation*}
    \begin{tikzcd}
      \CApre \ar[r, hook] \ar[d, "\tilde{\pi}"{swap}] \ar[dr, phantom, "\pullback", very near start]
      & \Coalg(C) \ar[d, "\overline{\pi}"] \\
      \EM(\cowriterfunc) \ar[r, hook]
      & \Coalg(\cowriterfunc)
    \end{tikzcd}
  \end{equation*}
We call $\gamma \from X \to CX$ \emph{free} if the underlying action $\tilde{\pi}(\gamma)$ is free, that is, if $\tilde{\pi}(\gamma)(x)$ is injective for all $x \in X$.
  In this case, $I^{\gamma(x)} \cong \intHom{N}{S}$ for all $x \in X$ and we say that $\gamma$ is \emph{uniform} if there is a map $u \from N \to S$ such that for every~$x\in X$, $\gamma(x) = (a, f)$ with $j^{a}(f) = u$.
\end{definition}

One interesting aspect of this definition is that it contains a general definition of homomorphism, the pre-cellular morphisms, for cellular automata which, to the best of our knowledge, has been absent from the literature so far.
However, pre-cellular morphisms are insufficient for relating the global behaviour of cellular automata.
In \cref{sec:sim}, we will refine them to \emph{cellular morphisms} and analyse the ensuing notion of behavioural equivalence.

We may sometimes need access to the spatial and the local rule components of cellular automata.
Given $\gamma \from X \to CX$ in $\CApre$, we define
\begin{align*}
  \gamma_{1} & \from X \to \cowriterfunc X
  & & \gamma_{2} : \prod_{x \in X} \intHom*{I^{\gamma_{1}(x)}}{S}
  \\
  \gamma_{1} & = \tilde{\pi}(\gamma)
  & & \gamma_{2}(x) = f, \text{ where } \gamma(x) = (a, f) \, .
\end{align*}
These definitions show the benefit of the coalgebraic approach, which hides the complex dependency of the type of $\gamma_{2}$ on $\gamma_{1}$.
In this notation, \cref{def:ca} means that $\CApre$ consists of all coalgebras $\gamma$ for which $\gamma_{1}$ is an Eilenberg-Moore coalgebra.

In order to understand the behaviour of cellular automata, we will introduce a global transition structure on the state of cells.
To this end, we need the \emph{evaluation map}~$e\colon CS\to S$ defined for all pairs~$(a,f)$ by
\begin{equation}
  \label{eq:evaluation-map}
  e(a, f) = f(a \comp i) \, ,
\end{equation}
where $i$ is the inclusion $N \hra M$.
This is well-defined  because if $(m, n) \in E^{a}$, then by definition $(a \comp i)(m) = (a \comp i)(n)$ and so $a \comp i$ sits in $I^{a}$.
\begin{definition}
  \label{def:configuration-global-rule}
Let $\gamma \from X \to CX$ be a cellular automaton.
A \emph{configuration} for $\gamma$ is a map $c \from X \to S$.
We denote by $\dual{X}$ the set $\intHom{X}{S}$ of all configurations.
  The \emph{global rule} for $\gamma$ is the map $G_{\gamma} \from \dual{X} \to \dual{X}$ given by the following composition.
  \begin{equation*}
    G_{\gamma}(c) = X \xrightarrow{\gamma} CX \xrightarrow{Cc} CS \xrightarrow{e} S
  \end{equation*}
\end{definition}
If we unfold the definition and omit the inclusion $t_{c}^{\gamma_{1}(x)}$, then we find that
\begin{equation*}
  G_{\gamma}(c)(x) = \gamma_{2}(x)(c \comp \gamma_{1}(x) \comp i) \, .
\end{equation*}
Intuitively,~$c$ evolves under $G_{\gamma}$ by looking up the current states of cells in the neighbourhood of $x$ and then applying the local rule $\gamma_{2}(x)$ in order to obtain the new state of $x$.

As the internal hom~$[-, -]$ is a bifunctor~$\op{\Set}\times\Set\to\Set$, 
taking configurations yield a functor $\dual{(-)} \from \op{\Set} \to \Set$ that is given by on objects by $\dual{X} = \intHom{X}{S}$ and on maps $f\colon X\to Y$ by pre-composition:
\begin{equation}
\label{eqn:dual}
\dual{f}\colon [Y, S]\to [X, S],\qquad	\dual{f}(c) = (X\xra{f} Y\xra{c} S)
\end{equation}
Composing this with the forgetful functor $\U{-} \from \CApre \to \Set$, we obtain the configurations functor $\dual{\U{-}} \from \op{(\CApre)} \to \Set$.
Following our convention, we will typically write $\dual{h}$ instead of $\dual{\U{h}}$.
With this notation set up, we obtain that pre-cellular morphisms preserve the behaviour of the global rule.
\begin{lemma}
  \label{lem:global-rule-natural}
  The global rule is a natural transformation $G_{-} \from \dual{\U{-}} \to \dual{\U{-}}$, that is, for all cellular morphisms $h \from \gamma \to \delta$ the following diagram commutes.
  \begin{equation*}
    \begin{tikzcd}
      \dual{\U{\gamma}} \rar{G_{\gamma}}
      & \dual{\U{\gamma}}
      \\
      \dual{\U{\delta}} \uar{\dual{h}} \rar{G_{\delta}}
      & \dual{\U{\delta}} \uar[swap]{\dual{h}}
    \end{tikzcd}
  \end{equation*}
\end{lemma}

The following examples illustrate how our definition captures both classical as well as non-uniform cellular automata by instantiation.

\begin{example}
\label{example:ca}
We discuss cases where the spatial layout is determined by actions of a freely generated monoid, by a free monoid action, and by a trivial action.
\begin{enumerate}
\item
\label{example:basic_general_ca}
Given an alphabet, we can view the transition structure of a deterministic automaton as a monoid action.
For example, let $M$ be the free monoid of finite words over $\{\leftelem, \rightelem\}$ with concatenation as monoid operation and the empty word $\monid$ as unit.
The action $\gamma_{1}$ of $M$ on~$X = \{w, x, y, z\}$ is generated from the transition graph in \cref{fig:intro_example_general_ca} just as words are read by a deterministic automaton.
For example, the map~$\gamma_1(w)\colon M\to X$ at cell~$w$ fulfills
\begin{equation*}
\gamma_1(w)(\leftelem) = x = \gamma_1(w)(\rightelem)\text{ and }\gamma_1(w)(\leftelem\cdot\rightelem) = y.
\end{equation*}
\begin{figure}[ht]
  \centering
  \begin{tikzpicture}[every state/.style={inner sep=2pt, minimum size=0pt}, shorten >=1pt]
    \node[state] (w) {$w$};
    \node[state, right=of w] (x) {$x$};
    \node[state, right=of x] (y) {$y$};
    \node[state, right=of y] (z) {$z$};
    \path[->]
    (w) edge node[above] {$\ell, r$} (x)
    (x) edge [loop above] node {$\ell$} ()
        edge [bend left] node[above] {$r$} (y)
    (y) edge [bend left] node[below] {$\ell$} (x)
        edge node[above] {$r$} (z)
    (z) edge [loop right] node {$\ell, r$} ()
    ;
  \end{tikzpicture}
  \vspace*{-2ex}
   \caption{A cellular automaton, where transitions are freely generated over $\set{\leftelem, \rightelem}$.}
\label{fig:intro_example_general_ca}
\end{figure}
For the neighborhood $N = \{\monid, \leftelem,\rightelem\}$ and state set $S = \swhiteblack$, the set~$I^{\gamma_1(w)}$ of local configurations~$f\from N \to S$ at cell~$w$ can be inferred from \cref{fig:intro_example_general_ca}. 
First,
\begin{align*}
E^{\gamma_1(w)} &= \setDef{(n,m)\in N\times N}{\gamma_1(w)(i(n))  = \gamma_1(w)(i(m))}  \\
			    &= \set{(\ell, r), (r,\ell)}\cup\set{(\ell, \ell), (r, r), (\monid, \monid)}.
\end{align*}
It follows that~$I^{\gamma_1(w)} = \setDef{f\in [N, S]}{f(\ell) = f(r)}$.   
Note that~$\gamma_1(y)$ is injective on~$N$ so~$I^{\gamma_1(y)}\cong \intHom{N}{S}$ by~\cref{lem:orb-inv}.
Computing the sets of orbits and orbit-invariant maps at the remaining cells is similarly straightforward:
 \begin{equation*}
 I^{\gamma_{1}(x)} = \setDef{f}{f(\ell) = f(\monid)}\text{ and }
 I^{\gamma_{1}(z)} = \setDef{f}{f(\ell) = f(\monid) = f(r)}.
 \end{equation*}
 Now, the local rule~$\gamma_2(v)$ of cell $v\in X$ is a map~$I^{\gamma_1(v)}\to S$.
  If we define them as below, we can iteratively compute a typical trace sequence as in \cref{table:intro_example_trace_seq}.
  \begin{gather*}
    \gamma_2(w)(k) \bydef
    \begin{cases}
      \sblack & \text{if $k(r) = \sblack,k(\monid)=\swhite$}   \\
      \swhite & \text{otherwise}
    \end{cases}
    \quad
    \gamma_2(z)(k) \bydef
    \begin{cases}
      \sblack & \text{if $k(\monid) = \swhite$} \\
      \swhite & \text{otherwise}
    \end{cases}
    \\[7pt]
    \gamma_2(x)(k) = \gamma_2(y)(k) \bydef
    \begin{cases}
      \sblack & \text{ if $k(\monid) = \sblack \text{ or } k(\rightelem)=\sblack$}
      \\
      \swhite & \text{otherwise}
    \end{cases}
  \end{gather*}
 \item
By letting the additive monoid~$(\Z^d, +, 0)$, where~$d\in\N$, act freely on itself, and by a choosing a neighbourhood format~$N\subseteq\Z^d$ and state set~$S = \2$, we recover a generalisation of \emph{binary $d$-dimensional cellular automaton}~\cite{Bhattacharjee20} in which the local rules of individual cells may differ across~$\Z^d$.
\item
Let~$\gamma$ be a cellular automaton with~$\gamma_1(x)(m) = x$ for every cell~$x$ and every~$m\in M$,~i.e.~the underlying monoid action of~$\gamma$ is the trivial one.
Then a map~$f\colon N\to S$ is orbit invariant precisely when it is constant.
\end{enumerate}
\end{example}

 \begin{table}[ht]
    \centering
    \begin{NiceTabular}{w{c}{1.5cm}|cccc}
      \RowStyle[bold]{}
      \rule[-2mm]{0pt}{0.5cm} % Force height of row and keep 4mm space below in other columns
      \diagbox{step}{cell} & $w$ & $x$ & $y$ & $z$ \\ \hline
      % \begin{tabular}{c|cccc}
      %   configuration & \multicolumn{4}{c}{cell state} \\
      %   & $w$ & $x$ & $y$ & $z$ \\ \hline
      0 & $\sdead$ & $\sdead$ & $\sdead$ & $\sdead$ \\
      1 & $\sdead$ & $\sdead$ & $\sdead$ & $\salive$ \\
      2 & $\sdead$ & $\sdead$ & $\salive$ & $\sdead$ \\
      3 & $\sdead$ & $\salive$ & $\salive$ & $\salive$ \\
      4 & $\salive$ & $\salive$ & $\salive$ & $\sdead$ \\
      5 & $\sdead$ & $\salive$ & $\salive$ & $\salive$ \\
      % $c_6$ & $\salive$ & $\salive$ & $\salive$ & $\sdead$ \\
      % $c_7$ & $\sdead$ & $\salive$ & $\salive$ & $\salive$ \\
      % $\vdots$ & $\vdots$ & $\vdots$ & $\vdots$ & $\vdots$
    \end{NiceTabular}
\caption{Initial segment of a configuration sequence of the CA in~\cref{example:ca}.\ref{example:basic_general_ca}.}
\label{table:intro_example_trace_seq}
\end{table}

\section{Functional behavioural equivalence}
\label{sec:sim}

In \cref{def:ca}, we established a category of cellular automata and pre-cellular morphisms.
The reason why we call these \emph{pre-}cellular morphisms stems from the fact that we cannot easily push configurations forward along a pre-morphism $f \from X \to Y$ between cellular automata $\gamma \from X \to CX$ and $\delta \from Y \to CY$:
if $f$ maps two cells $x$ and $x'$ in $X$ to the same cell $y \in Y$, then we are unable to decide whether $y$ should inherit the state from $x$ or $x'$.
Similarly, if $f$ is not surjective, we have cells in $Y$ with no pre-image under~$f$ and thus no states to chose from in a configuration.
This situation is analogous to that of vector fields on smooth manifolds which can, in general, not be pushed forward along a smooth map.
Of course, if $f$ is an isomorphism of cellular automata, this problem is easily remedied.
However, this solution rules out reasonable situations where $f$ is a quotient map or an embedding, while configurations can still be given to all cells in~$\delta$.
This leads us to the following definition.

\begin{definition}
\label{def:ca-morphism}
A \emph{cellular morphism} $h \from \gamma \to \delta$ between cellular automata is a pair $(f, s)$, where $f$ is a pre-cellular morphism and $s \from \dual{\U{\gamma}} \to \dual{\U{\delta}}$ is a section of the configuration mapping $\dual{f}:\dual{\U{\delta}}\to\dual{\U{\gamma}}$ from~\cref{eqn:dual}. That is,~$\dual{f} \comp s = \id$.
\end{definition}

In the following, a functor~$F\colon\BC\to\BD$ is \emph{conservative} if it reflects isomorphisms:
a morphism~$g$ in~$\BC$ is an isomorphism whenever~$Fg$ is.

\begin{lemma}
\label{lem:ca-category}
Cellular automata and cellular morphisms form a category $\CA$, which comes with a conservative, identity-on-objects functor $P \from \CA \to \CApre$.
\end{lemma}

The main motivation for cellular morphisms is that they allow us to push forward configurations.
This allows us to simulate the global behaviour of a cellular automaton in the codomain of a cellular morphism, as the following lemma shows.
\begin{lemma}[Naturality of global rules on cellular morphisms]
  \label{lem:global-rule-natural-full}
  For all cellular morphisms $(f, s) \from \gamma \to \delta$, we have $\dual{\U{f}} \comp G_{\delta} \comp s = G_{\gamma}$.
  % The following diagram commutes for all cellular morphisms  $(f, s) \from \gamma \to \delta$.
  % \begin{equation*}
  %   \begin{tikzcd}
  %     \dual{\U{\gamma}} \dar[swap]{s} \rar{G_{\gamma}}
  %     & \dual{\U{\gamma}}
  %     \\
  %     \dual{\U{\delta}} \rar{G_{\delta}}
  %     & \dual{\U{\delta}} \uar[swap]{\dual{f}}
  %   \end{tikzcd}
  % \end{equation*}
\end{lemma}

(Pre-)Cellular morphisms also reflect the behaviour of a cellular automaton back into the domain in the following sense.
Let $f \from \gamma \to \delta$ be a pre-cellular morphism.
We denote by $\im f \subseteq \U{\delta}$ the image of $f$ with inclusion $i_{f}$ and corestriction $e_{f} \from X \to \im f$ of $f$, such that $f$ factorises as $f = i_{f} \comp e_{f}$.
For $y \in \im f$, which means that $y = f(x)$ for some $x \in X$, we have for all $m \in M$ that
\begin{equation*}
  \delta_{1}(y)(m) = \delta_{1}(f(x))(m) = f(\gamma_{1}(x)(m)) \in \im f
\end{equation*}
and we can thus restrict $\delta$ to $\delta' \from \im f \to C(\im f)$.
The map $e_{f}$ is then also immediately a pre-cellular morphism $\gamma \to \delta'$.
Since $e_{f}$ is surjective, we can choose some $g \from \im f \to X$, which induces a pullback $\dual{g} \from \dual{X} \to \dual{(\im f)}$.
With this notation set up, we can prove that the behaviour of the cells in the image of $f$ is reflected in $\gamma$.
\begin{lemma}
  \label{lem:behaviour-reflection}
  Let $f \from \gamma \to \delta$ be a pre-cellular morphism, $X \xrightarrow{e_{f}} \im f \xrightarrow{i_{f}} Y$ its image factorisation as a map of sets and $\delta' \from \im f \to C(\im f)$ the induced cellular automaton.
  For all sections $g$ of $e_{f}$, we have that $\dual{g} \comp G_{\gamma} \comp \dual{e_{f}} = G_{\delta'}$.
  % For all sections $g$ of $e_{f}$, the following diagram commutes.
  % \begin{equation*}
  %   \begin{tikzcd}
  %     \dual{\U{\gamma}} \rar{G_{\gamma}}
  %     & \dual{\U{\gamma}} \dar{\dual{g}}
  %     \\
  %     \dual{\U{\delta}} \rar{G_{\delta'}} \uar{\dual{e_{f}}}
  %     & \dual{\U{\delta}}
  %   \end{tikzcd}
  % \end{equation*}
\end{lemma}

\cref{lem:global-rule-natural-full,lem:behaviour-reflection} show that cellular morphisms provide the right notion of behavioural equivalence for cellular automata, as they preserve and reflect the global behaviour.

% \begin{figure}[H]
%     \centering
%     \includegraphics[scale=1.2]{figures/hm_overview_v3.pdf}
%     \caption{
%         Summary of the approach of the correspondence between categorical
%         similarity and logical similarity:
%         the subcategory of $\pCA$
%         (which are left-action preserving (graph preserving) maps)
%         between CA cells that commute
%         with the local rules ($\CellSims$, categorical similarity)
%         is equal to the subcategory
%         of maps that preserve logical validity
%         ($\LogEmbs$, logical similarity).}
%     \label{fig:hm_overview}
% \end{figure}

\begin{example}
  For this example, we let $S$ be the integers modulo $2$, that is, $S = \Z_2$.
  The monoid will be $(\Z, 0, +)$ and the neighbourhood $N = \set{-1, 0, 1}$.
  We define one CA $\gamma \from \Z \to C\Z$, where the action $\gamma_{1}$ acts freely by addition and the local rule is given uniformly by addition modulo $2$: $\gamma_{2}(x)(v) = v(-1) + v(0) + v(1)$.
  The second CA we define is $\delta \from \Z^{2} \to C \Z^{2}$ on a grid.
  Its first component acts by addition on the first axis, that is, by $\delta_{1}(x, y)(k) = (x + k, y)$.
  The local rule is also given uniformly by $\delta_{2}(x, y)(v) = v(-1) + v(0) + v(1)$, which also uses addition modulo $2$.
  We can define a pre-cellular morphism $f \from \gamma \to \delta$ by simply embedding it on the line through the origin by setting $f(x) = (x, 0)$.
  It is straightforward to check that this is a pre-cellular morphism.

  In order to make this a morphism, we need to push forward configurations from $\gamma$ to $\delta$.
  We do this via the map $s \from \dual{\Z} \to \dual{\parens*{\Z^{2}}}$, which is defined as follows.
  \begin{equation*}
    s(c)(x, y) =
    \begin{cases}
      c(x), & y = 0 \\
      0, & y \neq 0
    \end{cases}
  \end{equation*}
%Since~$s(c)(f(x)) = s(c)(x, 0) = c(x)$,
This makes~$s$ into a section of $\dual{f}$ so that~$(f, s)$ is a cellular morphism $\gamma \to \delta$.
In particular,~\Cref{lem:global-rule-natural-full} entails that the global behaviour of $G_{\gamma}$ 
is preserved in $G_{\delta}$, see \cref{fig:cell_sim_tape_to_grid}.
\begin{figure}[ht]
\centering
\includegraphics[scale=0.75]{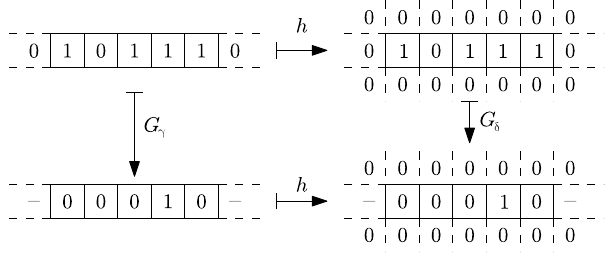}
\caption{
A cellular morphism~$h$ between a cellular automaton on~$\Z$ CA with states~$\Z/2\Z$
and a cellular automaton on~$\Z^2$:
it embeds the 1-dimensional grid as a row in the 2-dimensional grid.}
\label{fig:cell_sim_tape_to_grid}
\end{figure}
 Finally, $f$ is injective and thus $\im f \cong \Z$, which means that we not only reflect the global behaviour via \cref{lem:behaviour-reflection}, but also that $\delta' \from \im f \to C(\im f)$ is isomorphic to $\gamma$ in $\CA$.
\end{example}

In order to interpret a logic over cellular automata, it will be convenient to work with cellular automata that have a chosen base point.

\begin{definition}
\label{def:based-reachable-ca}
A \emph{based cellular automaton} is a pair $(\gamma, x_{0})$ with $x_{0} \in \U{\gamma}$ and a \emph{based cellular morphism} $(\gamma, x_{0}) \to (\delta, y_{0})$ is a cellular morphism $(f, s) \from \gamma \to \delta$ such that $f(x_{0}) = y_{0}$.
We call $(\gamma, x_{0})$ \emph{reachable} if the map $\gamma_{1}(x_{0})\colon M\to\U{\gamma}$ is surjective.
The category of based cellular automata and morphisms is denoted by $\pCA$ and that of reachable ones by $\prCA$.
\end{definition}

Our concept of reachability agrees with existing notions in the coalgebraic literature; 
see, e.g.,~Ad\'{a}mek et al.~\cite[Remark 3.14]{AdamekEA12} and Wi{\ss}mann~\cite[Inst.~4.4]{Wissmann21}.
Homomorphisms between based reachable cellular automata are automatically surjective.
Thus, \cref{lem:behaviour-reflection} ensures that the behaviour of all cells is reflected.

An interesting special case arises from pairs of based pre-morphisms between reachable cellular automata which, by \cref{lem:ca-category}, induce a cellular isomorphism.
\begin{proposition}
  \label{prop:pair-based-pre-morphisms-gives-iso}
  Let $f \from (\gamma, x_{0}) \to (\delta, y_{0})$ be a based pre-cellular morphism between reachable cellular automata.
  If there is a based pre-cellular morphism $g \from (\delta, y_{0}) \to (\gamma, x_{0})$, then $f$ is a cellular isomorphism with inverse $g$.
\end{proposition}
%
%%%%%%%%%%%%%%%%%%%%%%%%%%%%%%%%%%%%%%%%%%%%%%%%%
%%%%%%%%%%%%%%%%%%%%%%%%%%%%%%%%%%%%%%%%%%%%%%%%%
%%%%%%%%%%%%%%%%%%%%%%%%%%%%%%%%%%%%%%%%%%%%%%%%%
%%%%%%%%%%%%%%%%%%%%%%%%%%%%%%%%%%%%%%%%%%%%%%%%%
%%%%%%%%%%%%%%%%%%%%%%%%%%%%%%%%%%%%%%%%%%%%%%%%%
%%%%%%%%%%%%%%%%%%%%%%%%%%%%%%%%%%%%%%%%%%%%%%%%%
%%%%%%%%%%%%%%%%%%%%%%%%%%%%%%%%%%%%%%%%%%%%%%%%%
%Section:Modal logic
%
\section{Relational behavioural equivalence}\label{sec:relation-equivalence}
We proceed to describe a concept of cellular bisimulation relation between cells.
Typically, we think of a relation on sets~$X_1,X_2$ as a subset~$R\subseteq X_1\times X_2$. 
Under this view, relations comes naturally equipped with a pair of projection maps~$q_i\colon R\to X_i$ defined by~$q_i(x_1, x_2) = x_i$ resulting in the following \emph{span} in~$\Set$:
\begin{equation*}
X_1\xleftarrow{q_1} R \xra{q_2} X_2.
\end{equation*}
More generally, we view a relation in a category~$\BC$ as a span~$X\xleftarrow{q_1} R \xra{q_2} Y$. 
This affords the following notion of cellular bisimulation (cf.~\cite{staton_bisim}):

\begin{definition}
  A \emph{pre-cellular bisimulation} between cellular automata $\gamma$ and $\delta$ is a span $\gamma \xleftarrow{q_{1}} \rho \xrightarrow{q_{2}} \delta$ of pre-cellular morphisms in $\CApre$, which we denote by $(\rho, q) \from \proMor{\gamma}{\delta}$.
  Extending this, a \emph{cellular bisimulation} is a tuple $(\rho, q, R, p, g)$ consisting of a pre-cellular bisimulation $(\rho, q)$, a span $\dual{\U{\gamma}} \xleftarrow{p_{1}} R \xrightarrow{p_{2}} \dual{\U{\delta}}$ in $\Set$ and a map $g \from R \to R$, such that the following diagrams commute.
  \begin{equation}
    \label{eq:cellular-bisimulation-diagrams}
    \begin{tikzcd}
      R \rar{p_{1}} \arrow[d, "p_2" ']   %\dar{p_{2}}
      & \dual{\U{\gamma}} \dar{\dual{q_{1}}}
      \\
      \dual{\U{\delta}} \rar{\dual{q_{2}}}
      & \dual{\U{\rho}}
    \end{tikzcd}
    \qquad
    \begin{tikzcd}[column sep=large]
      R \arrow[d, "g" '] \rar{\pair{p_{1}}{p_{2}}} & \dual{\U{\gamma}} \times \dual{\U{\delta}} \dar{G_{\gamma} \times G_{\delta}} \\
      R \rar{\pair{p_{1}}{p_{2}}} & \dual{\U{\gamma}} \times \dual{\U{\delta}}
    \end{tikzcd}
  \end{equation}
\emph{Based pre-cellular bisimulations} are defined analogously as spans in $\pCApre$.
\end{definition}

It was identified early on that weak preservation of pullbacks is a salient property for the type functor of coalgebras~\cite{Rutten00}. 
A functor $F \from \BC\to\BC$ is said to \emph{weakly preserve pullbacks} if it maps pullbacks to \emph{weak pullbacks}.
Weak pullbacks are commuting squares as below, such that for every cospan $X \xleftarrow{w_1} W \xra{w_2} Y$ with $f\comp w_1 = g\comp w_2$, there exists a morphism~$w\colon W\to P$, not necessarily unique, such that~$w_i = q_i\comp w$.
\begin{equation*}
\begin{tikzcd}
P\dar[swap]{q_1} \rar{q_2} &Y\dar{g}  \\
X\rar{f} 				 &Z
\end{tikzcd}
\end{equation*}

\begin{lemma}
 \label{lem:weak-pb-preservation}
 The cellular automaton functor $C$ weakly preserves pullbacks.
\end{lemma}

The following example shows how we can obtain a (pre-)cellular bisimulation from a (pre-)cellular morphism.
\begin{example}
  \label{ex:cellular-bisimulation-from-morphism}
  Let $f \from \gamma \to \delta$ be a pre-cellular morphism and form the following pullback to obtain the graph of $f$.
  \begin{equation*}
    \begin{tikzcd}
      \mathrm{Gr}_{f} \rar{q_{2}} \arrow[d, "{q_1}" '] %\dar{q_{1}}
      \ar[dr, phantom, "\pullback", very near start]
      & \U{\delta} \dar{\id} \\
      \U{\gamma} \rar{f}
      & \U{\delta}
    \end{tikzcd}
  \end{equation*}
  Since $C$ weakly preserves pullbacks, see \cref{lem:weak-pb-preservation}, and monoid actions are closed under pullbacks, there is a cellular automaton structure $\rho_{f}$ on $\mathrm{Gr}_{f}$ and $(\rho_{f}, q)$ forms a pre-cellular bisimulation.

  Next, we define $R = \mathrm{Gr}_{\dual{f}}$ with the projections $p_{1}$ and $p_{2}$ as in the following pullback diagram in~$\Set$.
  \begin{equation}
   \label{eq:graph}
    \begin{tikzcd}
      \mathrm{Gr}_{\dual{f}} \rar{p_{1}} \arrow[d, "p_2" '] %\dar{p_{2}}
      \ar[dr, phantom, "\pullback", very near start]
      & \dual{\U{\gamma}} \dar{\id} \\
      \dual{\U{\delta}} \rar{\dual{f}}
      & \dual{\U{\gamma}}
    \end{tikzcd}
  \end{equation}
Note here the reversal of the projections in comparison to the graph of $f$, which corresponds to formally reversing the graph relation in order to obtain a span $(R, p) \from \proMor{\dual{\U{\gamma}}}{\dual{\U{\delta}}}$.
By definition and functoriality of $\dual{(-)}$, we have that
  \begin{equation*}
    \dual{q_{2}} \comp p_{2}
    = \dual{(f \comp q_{1})} \comp p_{2}
    = \dual{q_{1}} \comp \dual{f} \comp p_{2}
    = \dual{q_{1}} \comp p_{1}.
  \end{equation*}
By naturality of $G$ (cf.~\Cref{lem:global-rule-natural}) and definition of $\mathrm{Gr}_{\dual{f}}$, we have
  \begin{equation*}
    \dual{f} \comp G_{\delta} \comp p_{2}
    = G_{\gamma} \comp \dual{f} \comp p_{2}
    = G_{\gamma} \comp p_{1}.
  \end{equation*}
  We thus obtain a unique map $g \from \mathrm{Gr}_{\dual{f}} \to \mathrm{Gr}_{\dual{f}}$ into the pullback~(\ref{eq:graph}) with
  \begin{equation*}
    \pair{p_{1}}{p_{2}} \comp g = (G_{\gamma} \times G_{\delta}) \comp \pair{p_{1}}{p_{2}} \,.
  \end{equation*}
Therefore, the tuple $(\rho_{f}, q, R, p, g)$ is a cellular bisimulation.
If $(f, s) \from \gamma \to \delta$ is a cellular morphism, then the following two diagrams commute and, using that~$R$ is defined as a pullback, we obtain unique maps $t_{1} \from \dual{\U{\gamma}} \to R$ and $t_{2} \from \dual{\U{\delta}} \to R$ with $p_{1} \comp t_{1} = \id$, $p_{2} \comp t_{1} = s$, $p_{1} \comp t_{2} = \dual{f}$ and $p_{2} \comp t_{2} = \id$.
  \begin{equation*}
    \begin{tikzcd}
      \dual{\U{\gamma}} \rar{\id} \dar[swap]{s}
      & \dual{\U{\gamma}} \dar{\id} \\
      \dual{\U{\delta}} \rar{\dual{f}}
      & \dual{\U{\gamma}}
    \end{tikzcd}
    \qquad
    \begin{tikzcd}
      \dual{\U{\delta}} \rar{\dual{f}} \dar[swap]{\id}
      & \dual{\U{\gamma}} \dar{\id} \\
      \dual{\U{\delta}} \rar{\dual{f}}
      & \dual{\U{\gamma}}
    \end{tikzcd}
  \end{equation*}
  That is, any pre-cellular morphism yields a cellular bisimulation, but only a cellular morphism gives that the relation is inhabited via sections of the projections.
\end{example}

The following example shows that every pre-cellular bisimulations induces a cellular bisimulation, with the caveat that the constructed relation on configurations may be empty.
\begin{example}
  \label{ex:consistent-configuration-bisimulation}
  Let $(\rho, q)$ be a pre-cellular bisimulation $\proMor{\gamma}{\delta}$.
  We construct the relation of \emph{consistent configurations} of $(\rho, q)$ as the following pullback.
  \begin{equation}
    \label{eq:consistent-config-pullback}
    \begin{tikzcd}
      \Cons \rar{\ConsProj_{1}} \dar[swap]{\ConsProj_{2}}
      \ar[dr, phantom, "\pullback", very near start]
      & \dual{\U{\delta}} \dar{\dual{q_{2}}} \\
      \dual{\U{\gamma}} \rar{\dual{q_{1}}}
      & \dual{\U{\rho}}
    \end{tikzcd}
  \end{equation}
  The first condition of a cellular bisimulation is fulfilled by definition.
  For the second condition, we note that the follow identity holds.
  \begin{align*}
    \dual{q_{1}} \comp G_{\gamma}\comp \ConsProj_{1}
    & = G_{\rho} \comp \dual{q_{1}} \comp \ConsProj_{1}
      \tag*{by \cref{lem:global-rule-natural}} \\
    & = G_{\rho} \comp \dual{q_{2}} \comp \ConsProj_{2}
      \tag*{cellular bisimulation} \\
    & = \dual{q_{2}} \comp G_{\delta}\comp \ConsProj_{2}
      \tag*{by \cref{lem:global-rule-natural}}
  \end{align*}
This induces, by the mapping property of the pullback in (\ref{eq:consistent-config-pullback}), a unique map $\ConsCoalg \from \Cons \to \Cons$ with
  $\pair{\ConsProj_{1}}{\ConsProj_{2}} \comp \ConsCoalg = (G_{\gamma} \times G_{\delta}) \comp \pair{\ConsProj_{1}}{\ConsProj_{2}}$,
  % \begin{equation*}
  %   \pair{\ConsProj_{1}}{\ConsProj_{2}} \comp \ConsCoalg = (G_{\gamma} \times G_{\delta}) \comp \pair{\ConsProj_{1}}{\ConsProj_{2}} \, ,
  % \end{equation*}
  as required.
  Hence $(\rho, q, \Cons, \ConsProj, \ConsCoalg)$ is a cellular bisimulation.
\end{example}

In fact, the relation of consistent configurations has a special status, akin to the view on bisimilarity as final coalgebras.
\begin{proposition}
  \label{prop:consistent-config-final}
  Given a pre-cellular bisimulation $(\rho, q) \from \proMor{\gamma}{\delta}$, the cellular bisimulation $(\rho, q, \Cons, \ConsProj, \ConsCoalg)$ is final among all extensions of $(\rho, q)$ to a cellular bisimulation:
  If $(\rho, q, R, p, g)$ is another cellular bisimulation, then there is a unique map $h \from R \to \Cons$ with $\ConsProj_{k} \comp h = p_{k}$ and $h \comp g = \ConsCoalg \comp h$.
\end{proposition}

%
%%%%%%%%%%%%%%%%%%%%%%%%%%%%%%%%%%%%%%%%%%%%%%%%%
%%%%%%%%%%%%%%%%%%%%%%%%%%%%%%%%%%%%%%%%%%%%%%%%%
%%%%%%%%%%%%%%%%%%%%%%%%%%%%%%%%%%%%%%%%%%%%%%%%%
%%%%%%%%%%%%%%%%%%%%%%%%%%%%%%%%%%%%%%%%%%%%%%%%%
%%%%%%%%%%%%%%%%%%%%%%%%%%%%%%%%%%%%%%%%%%%%%%%%%
%%%%%%%%%%%%%%%%%%%%%%%%%%%%%%%%%%%%%%%%%%%%%%%%%
%%%%%%%%%%%%%%%%%%%%%%%%%%%%%%%%%%%%%%%%%%%%%%%%%
%Section:Modal logic
%
\section{Modal logic over cellular automata}\label{sec:logic}
We proceed to develop a modal language~$\ML$ for specifying properties in the configuration traces of cellular automata obtained by (iterative) application of the global rule from a chosen initial configuration and base cell.
Thus, our language will effectively be interpreted in based and reachable cellular automata equipped with an initial configuration.
Our language involves a collection of \emph{spatial modalities} indexed by elements of the monoid~$M$ as well as a distinguished \emph{update modality}. 
Intuitively, spatial modalities allow access to the states of cells reachable from the given cells in the current configuration. 
In contrast, the update modality enables access to cell states in the configurations obtained by application of the global rule to the initial configuration. 
The combination of these modalities result in a language for expressing properties of the local and global behaviour of cellular automata.

The syntax of our modal language, which extends the standard syntax of modal logic~\cite{BRV01} by 
infinitary disjunctions and the update modality mentioned above, is defined in the following. 

\begin{definition}
\label{def:grammar}
Let~$\ML$ denote the proper class generated by the following grammar:
\begin{equation*}
  \varphi
  \grammarDef s
  \mid \neg \varphi
  \mid \infdisj_{i \in I} \varphi_i
  \mid \spacedia{m} \varphi
  \mid \timedia \varphi
\end{equation*}
where $s \in S$ is a state,~$m \in\mon$, and $I$ is an index set.
An element~$\varphi\in\ML$ is called a \emph{modal formula}. 
For each regular cardinal~$\kappa$ (i.e.~an infinite cardinal which is not
cofinal with any smaller cardinals), we write~$\ML_{\kappa}$ for the fragment of~$\ML$ obtained by restricting to index sets of cardinality strictly below~$\kappa$. We refer to~$\varphi\in\ML_{\kappa}$ as \emph{$\kappa$-ary modal formulas}. A \emph{modal theory} is a set~$\Phi$ of modal formulas.
\end{definition}

%We emphasize that our syntax is given by a proper class of modal formulas. 
%For each modal formula~$\varphi$, let~$\kappa_{\varphi}$ denote the least regular cardinal such that~$\varphi\in\ML_{\kappa_{\varphi}}$ (such a cardinal exists). 
%In particular, each modal theory~$\Phi$ sits in~$\ML_{\lambda}$ 
%where~$\lambda = \sup\{\kappa_{\varphi}\setsep \varphi\in\Phi\}$.

\begin{notation}
We employ the following shorthands in the examples:
\begin{equation*}
\bot:=\bigvee_{i\in\varnothing}\varphi_i,
\qquad\qquad
\bigwedge_{i\in I}\varphi_i := \lnot\infdisj_{i\in I}\lnot\varphi_i,
\qquad
\text{and}
\qquad
(\varphi\to\psi) := \lnot\varphi\lor\psi,
\end{equation*}
where~$\varnothing$ is the empty set. Note that~$\ML$ is non-empty as~$\bot\in\ML$.
\end{notation}

As indicated, modal formulas are meant to be interpreted in based cellular automata together with an initial configuration.
We detail the semantics of modal formulas below.

\begin{definition}
\label{def:sems}
  Let $\gamma \from X \to CX$ be a cellular automaton, $x \in X$ a cell and $c \in \dual{X}$ a configuration.
  We define a \emph{modal forcing relation} $\sems$ by iteration on formulas in $\ML$ as follows.
  \begin{equation*}
    \begin{aligned}
      & (\gamma, x, c)\sems s
      & \text{iff} &&& c(x) = s
      \\  	%state
      & (\gamma, x, c)\sems \lnot\varphi
      &\text{iff} &&& (\gamma, x, c)\not\sems\varphi
      \\
      & (\gamma, x, c)\sems\infdisj\varphi_i
      &\text{iff} &&& (\gamma, x, c)\sems\varphi_i~\text{for some}~i\in I
      \\
      & (\gamma, x, c)\sems\spacedia{m}\varphi
      &\text{iff} &&& (\gamma, \gamma_1(x)(m), c)\sems\varphi
      \\
      & (\gamma, x, c)\sems\timedia\varphi
      &\text{iff} &&& (\gamma, x, G_{\gamma}(c))\sems\varphi
    \end{aligned}
  \end{equation*}
  We say that a based cellular automaton $(\gamma, x_{0})$ forces $\phi$ for a configuration $c$ if $(\gamma, x_{0}, c) \sems \phi$ holds.
We say that $\varphi$ is \emph{valid} in $(\gamma, x_{0})$ if $\varphi$ is forced by every configuration $c$, in which case we write $(\gamma, x_{0}) \sems \varphi$.
Finally, two formulas are logically equivalent, written $\phi \equiv \psi$, if for all based cellular automata $(\gamma, x_{0})$, we have that $(\gamma, x_{0}) \sems \phi$ if and only if $(\gamma, x_{0}) \sems \psi$.
\end{definition}

Given a cellular automaton~$\gamma\colon X\to GX$ and a cell~$x$, we may view the monoid action underlying~$\gamma$ as a functional accessibility relation where, for each~$m\in M$,~$\gamma_1(x)(m)$ encodes the
unique ``$m$-successor" of cell~$x$. 
From this perspective, spatial modalities are interpreted along the induced relation in a given configuration in the spirit of Kripke-style semantics of basic modal logic.
The key feature of our logic lies in the update modality~$\timedia$: it accesses cell states after \emph{transitioning} into the configuration~$G_{\gamma}(c)$ via the global rule.

\begin{example}
  Recall the cellular automaton $\gamma$ from \cref{example:basic_general_ca} in \cref{example:ca} with cell set $X = \set{w, x, y, z}$.
  Let $\varphi$ be the formula $\timedia\spacedia{\rightelem}\sblack$, which expresses that the right neighbour of a cell holds state $\sblack$ after a one-step update of the global configuration.
  For the configuration $c_{0}$ from \cref{table:intro_example_trace_seq} and cell~$y \in X$, we find that $\phi$ is forced there:
  \begin{equation*}
    (\gamma, y, c_0) \sems \timedia\spacedia{\rightelem}\sblack
  \end{equation*}
This holds because the update of $c_{0}$ via $G_{\gamma}$ yields $c_{1}$ in the second row of~\cref{table:intro_example_trace_seq}, where we see that $z$, the right neighbour of $y$, indeed has state $\sblack$.
\end{example}

Since the semantics of spatial modalities are given in terms of the monoid action that underlies a cellular automaton, it follows that they compose via the monoid operation.

\begin{lemma}
  \label{lemma:space_diamonds_and_monoid_mult}
  For all $m,n\in M$ and all $\varphi\in\ML$, we have $\spacedia{m} \spacedia{n} \varphi\equiv\spacedia{m \monop n} \varphi$.
\end{lemma}

Next, we provide some examples of classes of cellular automata which are axiomatisable in the language $\ML$.
Here, a class $\K$ of cellular automata is \emph{axiomatisable} if there is a modal theory $\Phi$ such that $\gamma$ is in $\K$ if, and only if, $(\gamma, x_0) \sems \varphi$ for all cells~$x$ and all $\varphi\in\Phi$. 
We give a class of cellular automata which is not axiomatisable over~$\ML$ in~\cref{example:oneway}.

\begin{example}[Quiescence]
Let~$q\in S$ be a state, and let $K_{q}\colon N\to S$ denote the constant map defined by $K_{q}(n) = q$, which is orbit-invariant for any action.
The state $q$ is \emph{quiescent} in a cellular automaton $\gamma$ if $\gamma_{2}(x)(K_{q}) = q$ for every cell $x$.
The class of cellular automata for which $q$ is a quiescent state is axiomatisable by the following formula:
\begin{equation*}
\infconj_{n \in N} \spacedia{n} q\to\timedia q.
\end{equation*}
Indeed, the formula is validated by a cellular automaton precisely if the next state of a cell with constant $q$ neighbourhood is again $q$, see~\cref{lem:local-rule-test-formula}.
\end{example}

\begin{example}[Periodicity]
  A cellular automaton $\gamma$ is \emph{periodic}~\cite{KAOL08_periodicity} if there exists $p\in\N$ (the \emph{period}), such that the global rule returns to the same configuration every $p \in \N$ iterations, that is, if $G_{\gamma}^p = \id_{\dual{\U{\gamma}}}$.
  Periodicity of a based cellular automaton can be expressed as:
  \begin{equation*}
    (\gamma, x_{0}) \sems
    \infdisj_{p \in \N}
    \infconj_{m \in \mon}
    \infconj_{t \in \N}
    \infconj_{s \in S}
    \spacedia{m} \timedia^t(s \to \timedia^p s).
  \end{equation*}
\end{example}

\begin{example}[Nilpotency]
A CA $\gamma \from X \to CX$ is \emph{nilpotent} if there exists~$t \in \N$ and a configuration $c_\bot\colon X \to S$ such that $G_{\gamma}(c_{\bot})=c_{\bot}$ (i.e.~$c_{\bot}$ is a fixed point of the global rule) and~$G_{\gamma}^{t+k}(c) = c_\bot$ for all configurations $c$ and $k \in \N$. 
In other words,~$c_{\bot}$ is reached under all initial configurations after at most~$t$ applications of the global rule.
Kari \cite{KA92_nilpotent} showed that it is undecidable whether a CA on~$\Z$ is nilpotent, and that any fixed point of the global rule is a configuration that assigns every cell the same quiescent state $s_\bot \in S$.
For a based and reachable cellular automaton $(\gamma, x_{0})$, nilpotency can be expressed as a validity property:
  \begin{equation*}
    (\gamma, x_{0}) \sems \infdisj_{s_\bot \in S} \infdisj_{t \in \N}
    \timedia^t \left(\infconj_{m \in \mon} \infconj_{k \in \N}
      \spacedia{m}\timedia^k s_\bot \right).
  \end{equation*}
  See Kamilya and Kari~\cite{KK21} for further work on nilpotency.
\end{example}

\begin{example}
In this final example, we show how to recover familiar modalities from temporal logic.
In particular, we obtain ``globally'', ``eventually'' and ``until'' by the following definitions
(see,~e.g.~\cite{Emerson90}).
\begin{equation*}
\mathbf{G} \phi = \infconj_{n \in \N}\timedia^n \varphi, 
\quad
\mathbf{F} \phi = \infdisj_{n \in \N} \timedia^n \varphi, 
\quad\text{and}\quad
\psi \mathrel{\mathbf{U}} \phi
  = \infdisj_{n \in \N} ((\infconj_{0 \leq i \leq n} \timedia^i \psi) \land \timedia^{n+1} \varphi).
\end{equation*}
%\begin{itemize}
%\item
%$\mathbf{G} \phi = \infconj_{n \in \N}\timedia^n \varphi$
%\item
%$\mathbf{F} \phi = \infdisj_{n \in \N} \timedia^n \varphi$
%\item
%$\psi \mathrel{\mathbf{U}} \phi
%= \infdisj_{n \in \N} ((\infconj_{0 \leq i \leq n} \timedia^i \psi) \land \timedia^{n+1} \varphi)$
%\end{itemize}
\end{example}

Our expressiveness theorems in \cref{sec:hennessy_milner} use particular kinds of formulas to probe values of configurations and the behaviour of local rules.
The following lemmas make these formulas and their properties explicit.
\begin{lemma}
  \label{lem:configuration-test-formula}
  Let $\gamma \from X \to C X$ be a cellular automaton, $x \in X$, and $f \in I^{\gamma_{1}(x)}$.
  Then for all $c \in \dual{X}$, we have
  \begin{equation*}
    c \comp \gamma_{1}(x) \comp i = f
    \iff (\gamma, x, c) \sems \infconj_{n \in N} \spacedia{n} f(n)
  \end{equation*}
\end{lemma}

\begin{lemma}
 \label{lem:local-rule-test-formula}
 Let $\gamma \from X \to C X$ be a cellular automaton, $x \in X$, and $f \in I^{\gamma_{1}(x)}$.
 Then for all $c \in \dual{X}$ with $c \comp \gamma_{1}(x) \comp i = f$ and $s \in S$, we have that
  \begin{equation*}
    \gamma_{2}(x)(f) = s
    \iff (\gamma, x, c) \sems \parens*{\infconj_{n \in N} \spacedia{n} f(n)} \to \timedia s
  \end{equation*}
\end{lemma}

Intuitively, the formula $\infconj_{n \in N} \spacedia{n} f(n)$ in \cref{lem:configuration-test-formula} compares the state of all the neighbours against the corresponding value of $f$.
In \cref{lem:local-rule-test-formula}, we use the same formula to check that, if the neighbours of a given cell are in a particular state, then the given cell transits to state $s$.

Our next step is to relate the formulas validated by two cellular automata.
If we were to follow the usual approach for, say, labelled transition systems, then we would call cells logically equivalent if they validate the same formulas.
However, just as pre-morphisms could not push the global configurations of a cellular automaton forward to the codomain, logical equivalence will take a seemingly non-symmetric form, where configurations need to be transported in a way that truth of formulas is preserved and reflected.
\begin{definition}
\label{def:logical-transport}
Let $(\gamma, x_{0})$ and $(\delta, y_{0})$ be based cellular automata.
A \emph{(strong) logical transport} $t \from (\gamma, x_{0}) \logTrans (\delta, y_{0})$ is a map $t \from \dual{\U{\gamma}} \to \dual{\U{\delta}}$, such that for all formulas $\phi$ and all configurations $c \in \dual{\U{\gamma}}$ we have
  \begin{equation*}
    (\gamma, x_{0}, c) \sems \phi
    \quad \text{iff} \quad
    (\delta, y_{0}, t(c)) \sems \phi
  \end{equation*}
\end{definition}

\begin{lemma}
\label{lem:logical-transport-category}
The identity map is a logical transport and logical transports are closed under composition.
\end{lemma}

Using \cref{lem:logical-transport-category}, we can define a category $\LogTrans$ with based, reachable cellular automata as objects and logical transports as maps.
Our first important result is that moving along cellular morphism is sound for logical transports.
\begin{theorem}
  \label{thm:cellular-morphism-gives-logical-transport}
  If $(f, s)$ is a based cellular morphism $(\gamma, x_{0}) \to (\delta, y_{0})$, then $s$ is a logical transport $(\gamma, x_{0}) \logTrans (\delta, y_{0})$.
\end{theorem}
This allows us to define a functor $\prCA \to \LogTrans$ by projecting $(f, s)$ on $s$.

\Cref{thm:cellular-morphism-gives-logical-transport} can be proven by appealing to its relational counterpart in \cref{thm:bisimulation-to-logical-equivalence} and \cref{ex:cellular-bisimulation-from-morphism}.
Indeed, the following definition generalises logical transport from maps to relations.
\begin{definition}
  \label{def:logical-equivalence}
  Let $(\gamma, x_{0})$ and $(\delta, y_{0})$ be based cellular automata.
  A \emph{logical equivalence}, is a span $\dual{\U{\gamma}} \xleftarrow{p_{1}} R \xrightarrow{p_{2}} \dual{\U{\delta}}$, written $(R, p)$, such that the following holds for all $c \in R$ and all formulas $\phi$:
  \begin{equation*}
    (\gamma, x_{0}, p_{1}(c)) \sems \phi
    \iff
    (\delta, y_{0}, p_{2}(c)) \sems \phi.
  \end{equation*}
\end{definition}

The following theorem shows that cellular bisimulations induce logical equivalence.
It is proven by a straightforward induction on formulas.
\begin{theorem}
  \label{thm:bisimulation-to-logical-equivalence}
  If $(\rho, f, R, p, g)$ is a cellular bisimulation $\proMor{\gamma}{\delta}$, then $(R, p)$ is a logical equivalence between $(\gamma, f_{1}(u))$ and $(\delta, f_{2}(u))$ for all $u \in \U{\rho}$.
\end{theorem}

%Our main goal is to prove that this category is isomorphic to $\prCA$, thereby establishing a Hennessy-Milner style correspondence between logical and behavioural equivalence.

%%%%%%%%%%%%%%%%%%%%%%%%%%%%%%%%%%%%%%%%%%%%%%%%%
%%%%%%%%%%%%%%%%%%%%%%%%%%%%%%%%%%%%%%%%%%%%%%%%%
%%%%%%%%%%%%%%%%%%%%%%%%%%%%%%%%%%%%%%%%%%%%%%%%%
%%%%%%%%%%%%%%%%%%%%%%%%%%%%%%%%%%%%%%%%%%%%%%%%%
%%%%%%%%%%%%%%%%%%%%%%%%%%%%%%%%%%%%%%%%%%%%%%%%%
%%%%%%%%%%%%%%%%%%%%%%%%%%%%%%%%%%%%%%%%%%%%%%%%%
%%%%%%%%%%%%%%%%%%%%%%%%%%%%%%%%%%%%%%%%%%%%%%%%%
%Section:HM Theorem
%
\section{A Hennessy-Milner style expressiveness theorem}\label{sec:hennessy_milner}

In this section, we establish the full correspondence between cellular morphisms and cellular bisimulations on the one hand, and logical transport and logical equivalence on the other hand.
To be precise, we prove the reverse directions of \cref{thm:cellular-morphism-gives-logical-transport,thm:bisimulation-to-logical-equivalence} in \cref{theo:hm_theorem,thm:relational-hm} below.
This should be compared with the seminal result of Hennessy and Milner~\cite{HM85}: given a pair of states in finitely-branching labelled transition systems (of the same similarity type) which are not bisimilar, there is a modal formula which distinguishes them.

% The following theorem shows that this functor is an isomorphism, which is the sought correspondence between cellular morphisms and logical transports.

The proof of our theorems will make essential use of \cref{lem:local-rule-test-formula} to probe local rules.
This, however, will require that we have sufficient supply of configurations that are compatible with a given neighbourhood configuration.
\begin{lemma}
  \label{lem:lifting-local-configurations}
  Suppose that $S$ is inhabited, i.e., it has an element, and that $f \from N \to S$ and $g \from N \to X$ are maps.
  If $g(n) = g(n')$ implies $f(n) = f(n')$ for all $n \in N$, then there is a lifting $h$ that renders the following diagram commutative.
  \begin{equation*}
    \begin{tikzcd}
      N \rar{f} \dar[swap]{g} & S \\
      X \ar[ur, dashed, "h"{swap}]
    \end{tikzcd}
  \end{equation*}
\end{lemma}
In light of \cref{lem:lifting-local-configurations}, we will continue with the assumption that $S$ is inhabited.

On the functional side, that is when going from logical transports to cellular morphisms, we have to restrict to reachable cellular automata with a free action.
These conditions allow us to construct the pre-morphism between the given cellular automata only from equi-satisfiability of formulas.
\begin{theorem}\label{theo:hm_theorem}
Suppose that~$t\colon(\gamma, x_0)\logTrans(\delta, y_0)$ is a logical transport between based, reachable cellular automata with a free action $\gamma_1$.
Then there exists a unique based cellular morphism~$\bar{t}\colon(\gamma, x_0)\to(\delta, y_0)$ 
with~$\pi_2(\bar{t}) =t$.
\end{theorem}

Before discussing the proof of \cref{theo:hm_theorem}, we first state its relational counterpart, as we can use the latter in the proof.
This theorem is the direct analogue of the classical Hennessy-Milner theorem, as it establishes that logical equivalence for cellular automata is a cellular bisimulation relation.
\begin{theorem}
  \label{thm:relational-hm}
Let $(\gamma, x_{0})$ and $(\delta, y_{0})$ be based cellular automata and $(R, p)$ a logical equivalence between them with a section $s_{2}$ of $p_{2}$.
  Then there is a cellular bisimulation $(\rho, q, \Cons, p', g)$ with $(x_{0}, y_{0}) \in \U{\rho}$ and a map $h \from R \to \Cons$ with $p_{k}' \comp h = p_{k}$.
\end{theorem}

The proof of \cref{thm:relational-hm} relies only on the formulas in \cref{lem:configuration-test-formula,lem:local-rule-test-formula},  which in turn only employ conjunction over the neighbourhood.
Thus, the general infinitary disjunction of our logic is unnecessary to establish the correspondence between logical and behavioural equivalence, and so we may restrict it to only range over the neighbourhood.
We discuss this further in \cref{sec:conclusion}.

The existence of a section for $p_{2}$ in \cref{thm:relational-hm} is seemingly strong, but it is fulfilled in important case of \cref{theo:hm_theorem}, as we will see next.
We leave it for future work to identify potential improvements of this result.
\begin{proof}[\Cref{theo:hm_theorem}]
  We provide only an outline of the proof.
  Given a logical transport $t \from (\gamma, x_0) \logTrans (\delta, y_0)$, we will construct a based cellular morphism $\bar{t} \from (\gamma, x_0) \to (\delta, y_0)$ with $\pi_{2}(\bar{t}) = t$.
  This means that we have to show that there exists a (necessarily unique) based pre-cellular morphism $f_t \from (\gamma, x_0) \to (\delta, y_0)$, such that $t$ is a section of $\dual{f_t}$.
  For brevity, let us write $X = \U{\gamma}$ and $Y = \U{\delta}$.

  First, note that since $(\gamma, x_0)$ is based and reachable, every cell $x \in X$ has the form $\gamma_1(x_0)(m)$ for some $m\in M$.
  Since $\gamma_{1}$ is free, this $m$ is unique and we can define $f_{t}$ as the unique map with the following graph.
  \begin{equation*}
    \setDef{(x, \delta_{1}(y_{0})(m)) \in X \times Y}{m \in M, x = \gamma_{1}(x_{0})(m)}
  \end{equation*}
  Moreover, $f_t$ is equivariant by definition.
  One then shows that $t$ and $\dual{f_{t}}$ are inverses of each other, which follow from the definition of logical transport, reachability, and freeness of $\gamma_{1}$.
  This implies that the graph of $t$ has a section for its second projection.
  Moreover, it is a logical equivalence and \cref{thm:relational-hm} yields a cellular bisimulation $(\rho, q, \Cons, p', g)$ with $(x_{0}, y_{0}) \in \U{\rho}$.
  Finally, we define $r \from X \to \U{\rho}$ as follows, again appealing to $\gamma_{1}$ being free and to reachability.
  \begin{equation*}
    r(x) = \rho_{1}(x_{0}, y_{0})(m), \quad \text{for } x = \gamma_{1}(x)(m)
  \end{equation*}
  This map is a pre-cellular morphism and we have $q_{2} \comp r = f_{t}$, which means that $f_{t}$ is a pre-cellular morphism and the pair $(f_{t}, t)$ is a cellular morphism.
  \qed
\end{proof}

We mention two use cases of~\Cref{theo:hm_theorem}. First, we reemphasize that in order to prove that two cells are inequivalent under cellular bisimilarity, it suffices to exhibit a modal formula which distinguishes them, that is, a formula which is satisfied in one structure and falsified in the other.
In another direction, our result enables one to show that particular classes of cellular automata are \emph{not} axiomatisable by exploiting using standard techniques from model theory.
We illustrate the non-axiomatisability of one-way cellular automata in the following example, cf.~\cite{Roka94}.

\begin{example}[Inexpressability of one-way cellular automata]
\label{example:oneway}
A based and reachable cellular automaton~$(\gamma, x_0)$ is \emph{one-way} if, for each cell~$x$, 
$\gamma_1(x)(m) = x$ implies~$m =\monid$. 
For instance, the cellular automata~$(\gamma, w)$ of~\Cref{fig:intro_example_general_ca} is not one-way because
$\gamma_1(x)(r\monop\ell) = x$ and $r\monop\ell\neq\monid$.

For the sake of contradiction, suppose that~$\Phi$ axiomatises one-way cellular automata. 
Write~$\varphi$ for the conjunction over~$\Phi$.
It is enough to exhibit a one-way cellular automaton~$\gamma$ along with a cellular morphism~$\gamma\to\delta$ with~$\delta$ failing to be one-way since, by~\Cref{theo:hm_theorem} we will then have that~$(\gamma, x_0)\sems\varphi$ if and only if~$(\delta, y_0)\sems\varphi$.
This leads to a contradiction with the assumption that~$\varphi$ axiomatises one-way cellular automata.

To this end, let~$(\gamma, x_0)$ and~$(\delta, x_0)$ be the cellular automata on the additive monoid~$\Z_2$ and state set~$\{*\}$ with the monoid actions depicted below.
\begin{equation*}
\centering
\begin{tikzcd}[arrow style=tikz]
x_0
            \ar[r, bend right, "1" '] \ar[loop above, "0"]
&x_1
            \ar[l, bend right, "1" '] \ar[loop above, "0"]
\end{tikzcd}
\qquad\qquad
\begin{tikzcd}[arrow style=tikz]
x_0
            \ar[r, "1"] \ar[loop above, "0"]
&x_1 \ar[loop above, "{0,1}"]
\end{tikzcd}
\end{equation*}
The identity assignment is clearly a based cellular morphism~$(\gamma, x_0)\to(\delta, x_0)$.
Note that~$(\gamma, x_0)$ is one-way, but~$(\delta, x_0)$ is
not one-way because~$\delta_1(x_1)(1) = x_1$.
\end{example}

\subsection*{Scope of \cref{thm:relational-hm}}

In \cref{thm:relational-hm}, we additionally assumed that the projection $p_{2}$ of a logical equivalence has a section.
This is a non-trivial property.
It generally holds in the setting of \cref{theo:hm_theorem}, but one may wonder whether this conditions holds also for non-free systems.
The following example constructs a cellular bisimulation for non-free CA, such that \cref{thm:bisimulation-to-logical-equivalence} yields a logical equivalence for which the section condition is fulfilled.
Simultaneously, the example shows that \cref{thm:relational-hm} applies to CA that are neither free, reachable nor uniform.

We work with the monoid $(\N, +, 0)$ of natural numbers and states $S = \N$.
First, we let $X = \Z_{2} \times \Z_{6}$ and define the action $\gamma_{1}(n, m)(k) = (n, m + k \bmod 6)$, which looks like two copies of circles, see \cref{fig:bisimilar-sub-ca} on the left.
\usetikzlibrary {perspective}
\begin{figure}[ht]
  \centering
  \begin{tikzpicture}[x=0.8cm, y=0.8cm, 3d view]
    \draw[color=green!70!black] (0,0,0) circle (1);
    \draw(0,0,1) circle (1);
    % \draw(0,0,2) circle (1);
    \pgfmathsetmacro\n{6}
    \foreach \i in {0,1,...,5} {
      \pgfmathsetmacro\r{\i*(360/\n)};
      \fill ($(\r:1) + (0,0,0)$) circle (1pt) coordinate (c0-\i);
      \fill ($(\r:1) + (0,0,1)$) circle (1pt) coordinate (c1-\i);
      % \fill ($(\r:1) + (0,0,2)$) circle (1pt) coordinate (c2-\i);
      \pgfmathsetmacro\p{\r - 10};
      \node (n1-\i) at ($(\p:1.5) + (0,0,0)$) {\i};
    };
    \foreach \i in {0,1} {
      \draw[thick,red,->] ([shift=(10:1.5)]0,0,\i) arc (3:35:1.5) node[pos=0.7, above right=-0.1]{$+1$};
    };
    % \fill ($(c-5) + (0.1, -0.6)$) circle (1pt) coordinate (b);
    % \node at ($(b) + (-0.2, -0.2)$) {b};
    % \draw[->, shorten >= 2pt] (b) to (c-5);
  \end{tikzpicture}
  \qquad
  \begin{tikzpicture}[x=0.8cm, y=0.8cm]
    \draw[color=green!70!black] (0,0) circle (1);
    \pgfmathsetmacro\n{6}
    \foreach \i in {0,1,...,5} {
      \pgfmathsetmacro\r{\i*(360/\n)}
      \fill (\r:1) circle (1pt) coordinate (c-\i);
      \pgfmathsetmacro\p{\r + 10}
      \node (n-\i) at (\p:1.2) {\i};
    };
    \draw[thick,red,->] ([shift=(5:1.4)]0,0) arc (3:35:1.4) node[midway, right]{$+1$};
    \fill ($(c-0) + (0.5, -0.6)$) circle (1pt) coordinate (a);
    \node at ($(a) + (0.2, -0.2)$) {a};
    \draw[->, shorten >= 2pt] (a) to (c-0);
  \end{tikzpicture}

  \caption{Two CA with cellularly bisimilar subsystems coloured in green}
  \label{fig:bisimilar-sub-ca}
\end{figure}
Next, we set $Y = \Z_{6} + \set{a}$ and define an action by
\begin{align*}
  & \delta_{1}(n)(k) = n + k \bmod 6, \text{ if } n \in \Z_{6} \\
  & \delta_{1}(a)(k) =
    \begin{cases}
      a, & k = 0 \\
      k' \bmod 6, & k = k' + 1
    \end{cases}
\end{align*}
This can be pictured like a discrete circle with a lasso incoming at $0$ from $a$, see \cref{fig:bisimilar-sub-ca} on the right.
The local rules do not matter for the time being and can be chosen to be constant everywhere.

We define a relation $Q \subseteq X \times Y$ by
\begin{equation*}
  Q = \setDef{((0, n), n)}{n \in \Z_{6}}
\end{equation*}
with action $\rho_{1}((0, n), n)(k) = ((0, n + k), n + k)$ and constant local rules.
This makes the projections $q_{1} \from Q \to X$ and $q_{2} \from Q \to Y$ evidently pre-cellular morphisms and we thus obtain a pre-cellular bisimulation.

\Cref{ex:consistent-configuration-bisimulation} yields the canonical cellular bisimulation $(\rho, q, \Cons, \ConsProj, \ConsCoalg)$ of consistent configurations.
We use its definition as a pullback in order to construct a section for $\ConsProj_{2}$, by producing a map $t$ that renders the following diagram commutative.
\begin{equation*}
  \begin{tikzcd}
    \dual{X} \rar{t} \dar[swap]{\id}
    \ar[dr, phantom, "\pullback", very near start]
    & \dual{Y} \dar{\dual{q_{2}}} \\
    \dual{X} \rar{\dual{q_{1}}}
    & \dual{Q}
  \end{tikzcd}
\end{equation*}
This gives then, by the mapping property of pullbacks, a map $s_{2} \from \dual{X} \to \Cons$.
We define $t$ on $c \in \dual{X}$ and $u \in Y$ as follows.
\begin{equation*}
  t(c)(u) =
  \begin{cases}
    c(0, u), & u \in \Z_{6} \\
    0, & u = a
  \end{cases}
\end{equation*}
We then have for all $c \in \dual{X}$ and $((0, n), n) \in Q$ that
\begin{equation*}
  (\dual{q_{2}} \comp t)(c)((0, n), n)
  = t(c)(n)
  = c(0, n)
  = (\dual{q_{1}} \comp \id)(c)((0, n), n) \, .
\end{equation*}
This gives us $s_{2}$ with $\ConsProj_{2} \comp s_{2} = \id$ by the pullback property and, combined with \cref{thm:bisimulation-to-logical-equivalence}, $(\Cons, \ConsProj)$ is a logical relation that fulfils the conditions of \cref{thm:relational-hm}.
In other words, the conditions of \cref{thm:relational-hm} are fulfilled also in non-free and non-reachable cases, but it seems to be non-trivial.

Finally, we can equip the cellular automata also with non-uniform rules.
For instance, we can define
\begin{equation*}
  \gamma_{2}(n, m)(f) = (f(0) + f(1)) \bmod (m + 1)
\end{equation*}
and
\begin{align*}
  & \delta_{2}(a)(f) = f(0) + 1 \\
  & \delta_{2}(m)(f) = (f(0) + f(1)) \bmod (m + 1), \quad (k\in\Z_{6}) \, .
\end{align*}
It is clear that, since $Q$ relates only two circle-shaped subsystems of the two CA with precisely the same local rules, $Q$ remains a cellular bisimulation in this case.
This shows that \cref{thm:relational-hm} applies to CA that are neither free, reachable nor uniform.
Similarly, one can easily find examples that demonstrate that also \cref{theo:hm_theorem} applies in all these cases, except that the relation between two CA has to be given by a map instead of a relation.

%
%%%%%%%%%%%%%%%%%%%%%%%%%%%%%%%%%%%%%%%%%%%%%%%%%
%%%%%%%%%%%%%%%%%%%%%%%%%%%%%%%%%%%%%%%%%%%%%%%%%
%%%%%%%%%%%%%%%%%%%%%%%%%%%%%%%%%%%%%%%%%%%%%%%%%
%%%%%%%%%%%%%%%%%%%%%%%%%%%%%%%%%%%%%%%%%%%%%%%%%
%%%%%%%%%%%%%%%%%%%%%%%%%%%%%%%%%%%%%%%%%%%%%%%%%
%%%%%%%%%%%%%%%%%%%%%%%%%%%%%%%%%%%%%%%%%%%%%%%%%
%%%%%%%%%%%%%%%%%%%%%%%%%%%%%%%%%%%%%%%%%%%%%%%%%
%Section:Concluding remarks
%
\section{Conclusion and Future Work}
\label{sec:conclusion}

We have described a novel coalgebraic model of cellular automata that, unlike the traditional notion of cellular automata, permits non-uniformity in the local rules and neighbourhood structures of individual cells.
One of the primary benefits of our coalgebraic model is that it supports a robust concept of cellular morphism between cellular automata, which, to the best of our knowledge, has been absent from the literature until now.
% Cellular morphisms combine the native concept of coalgebra morphism with a map that can push global configurations of cells forward.
% Cellular morphisms provide good semantics of cellular behaviour as they preserve and reflect global behaviour.
We have introduced a modal logic for reasoning about the global behaviour of cellular automata.
Our language involves two types of unary modalities: spatial modalities indexed by monoid elements provide access to reachable cell states within a given configuration, while our update modality allows access to cell states after application of the global rule.
Thus, our language is meant to express dynamic properties of cellular automata such as periodicity or nilpotency.
We provide a Hennessy-Milner style theorems (\cref{theo:hm_theorem,thm:relational-hm}) which relate logical and behavioural equivalence.
In particular, we show that cells are identified under cellular morphisms precisely if they satisfy the same modal formulas.

We mention several directions for future work.
We first note that careful inspection of the construction of our coalgebraic type functor suggests that our model of cellular automata might be adapted to base categories with finite limits and Cartesian closed structure, or perhaps even symmetric monoidal closed structure.
This more general landscape includes, for example, the symmetric monoidal closed category of $1$-bounded metric spaces and non-expansive maps.
Extending our framework there might afford a quantitative comparison of cellular automata in the form of a behavioural pseudometric~\cite{BBK+18:CoalgebraicBehavioralMetrics}.
We would then also be able to capture linear or continuous cellular automata~\cite{Bhattacharjee20} by instantiation.

Our coalgebraic perspective on cellular automata also opens up the possibility of various extensions.
In particular, it should be possible to integrate computational effects by adding a monad component~\cite{Jacobs12:TraceSemanticsDeterminization,SBB+13:GeneralizingDeterminizationAutomata}.
This would give access to, for example, probabilistic behaviour that can be used to solve the leader election problem without identifiers~\cite{BFP+08:LeaderElectionAnonymous}.

An interesting aspect of the global rule is its definition in terms of an evaluation map. % given in \cref{eq:evaluation-map}.
This map can be seen as a distributive law of the constant functor over $C$ and then potentially related to bialgebras~\cite{Klin09:BialgebraicMethodsModal}.
However, evaluation functions also appear in coalgebraic approaches to behavioural metrics~\cite{BBK+18:CoalgebraicBehavioralMetrics} and it would be beneficial to understand if metrics and global behaviour are formally related.

Notions of behavioural equivalence have traditionally been characterised not only in dedicated modal logics, as we do in the current paper, but also in terms of behavioural equivalence games~\cite{Stirling99}.
We hope to employ recent work on coalgebraic equivalence games~\cite{FordEA22,KonigEA20} in the development of a game-theoretic perspective on cellular behavioural equivalence.

Finally, our logic features infinitary operations, which is not strictly necessary for completeness of logical equivalence as we discussed after \cref{thm:relational-hm}.
However, the examples in \cref{sec:logic} show that most interesting properties of cellular automata do require infinitary disjunction and conjunction if those properties are to be expressed by one formula.
With that being said, these properties could also be expressed by integrating recursion, in the style of the modal $\mu$-calculus or branching time logic, into our logic.
We would expect that the resulting logic would be sufficiently expressive for reasoning about long-term behaviour, just as modal $\mu$-calculi~\cite{BS07:ModalMuCalculi} are.
Thus, an enticing direction for future work is to devise fixed point logics for cellular automata in terms of our coalgebraic framework, which might require slight modifications to take the global behaviour into account.
This can give access to powerful model checking techniques~\cite{HHP+24:GenericModelChecking,KCM+24:CoalgebraicCTLFixpoint}.

%\begin{credits}
%\end{credits}

\bibliography{references.bib}
\clearpage
\appendix
%%%%%%%%%%%%%%%%%%%%%%%%%%%%%%%%%%%%%%
%%%%%%%%%%%%%%%%%%%%%%%%%%%%%%%%%%%%%%
%%%%%%%%%%%%%%%%%%%%%%%%%%%%%%%%%%%%%%
%%%%%%%%%%%%%%%%%%%%%%%%%%%%%%%%%%%%%%
%%%%%%%%%%%%%%%%%%%%%%%%%%%%%%%%%%%%%%
%%%%%%%%%%%%%%%%%%%%%%%%%%%%%%%%%%%%%%
%%%%%%%%%%%%%%%%%%%%%%%%%%%%%%%%%%%%%%
%%%%%%%%%%%%%%%%%%%%%%%%%%%%%%%%%%%%%%
%Appendix on preliminaries
%
\section{Details for~\cref{sec:prelim}}
%begin takeout
\takeout{
Fix a monoid~$(M, \monop, \monid)$. We briefly sketch the
details that the functor~$[M, -]\colon\Set\to\Set$ carries the
structure of a comonad with counit
\todo{Maybe we omit this since it is well known?}
\begin{equation}\label{eq:eps}
\eps\colon[M, X]\to X,\qquad (f\colon M\to X)\mapsto f(\monid)
\end{equation}
and comultiplication~$\delta\colon[M, [M, X]]\to [M, X]$ defined
for all~$f\colon M\to [M, X]$ by
\begin{equation}\label{eq:delta}
\delta(f)(m)(n) = f(n\monop m).
\end{equation}
This is just a recapitulation of the material from~\autoref{sec:prelim}.
For convenience, we also recall from~\cref{expl:exponent} that the
action of~$[M, -]$ on a map~$f\colon X\to Y$ is given by pre-composition:
\begin{equation}\label{eq:exponent-action}
[M, f](M\xra{m} X) = f\cdot m.
\end{equation}

\begin{lemma}
The mappings~$\eps$ and~$\delta$ are the components of
a natural transformation.
\end{lemma}
\begin{proof}
\emph{Naturality of~$\eps$}:
Our goal is to show that the following square commutes for
a map~$f\colon X\to Y$:
\begin{equation*}
\begin{tikzcd}[column sep = 35]
{[M, X]}\ar[d, swap, "\eps"]\ar[r, "{[M, f]}"]	&{[M, Y]}\ar[d, "\eps"] \\
      X \ar[r,  "f"]						& Y
\end{tikzcd}
\end{equation*}
This follows readily from associativity of composition and from
the definitions from~(\ref{eq:exponent-action}) and~(\ref{eq:eps}).
In detail, given a map~$k\colon M\to X$, we compute as follows:
\begin{align}
\eps_Y\cdot[M, f](k) &= \eps_Y\cdot(f\cdot k)	\tag{by~(\ref{eq:exponent-action})} \\
			       &= (f\cdot k)(\monid)		\tag{by~(\ref{eq:eps})} 		      \\
			       &= f(k(\monid))			\tag{by associativity}			       \\
			       &= f(\eps_X(k))			\tag{by~(\ref{eq:eps})}		      \\
			       &= f\cdot\eps_X(k)		\tag{by associativity}
\end{align}
We conclude that~$\eps$ is a natural transformation, as required.

~\newline
\noindent
\emph{Naturality of~$\delta$}: We will show that the following
square commutes for a map~$f\colon X\to Y$:
\begin{equation*}
\begin{tikzcd}[column sep = 45]
{[M, X]}\ar[d, swap, "\delta_X"]\ar[r, "{[M, f]}"]		&{[M, Y]}\ar[d, "\delta_Y"] 	\\
{[M, [M, X]]} \ar[r, "{[M, [M, f]]}"]			&{[M, [M, Y]]}
\end{tikzcd}
\end{equation*}
Indeed, given a map~$k \colon M \to X$ we compute as follows:
\begin{align*}
 (f^\mon)^\mon \circ \delta_X (k)  & = (f^\mon)^\mon (\lam{m}{\lam{n}{k(n\monop m)}}) 	\tag{by~(\ref{eq:delta})}	\\
           					  & = \lam{m}{f^\mon \circ \lam{n}{k(n \monop m)}}		\tag{} \\
  						  & = \lam{m}{\lam{n}{(f \circ k)(n \monop m)}} 			\tag{} \\
            					  & = \delta_Y(f \circ k)							\tag{} \\
            					  & = (\delta_Y \circ f^\mon)(k)						\tag{}
\end{align*}
We conclude that~$\delta$ is a natural transformation, as desired.
\end{proof}

\begin{lemma}
$\cowriterfunc\colon\Set\to\Set$ is a comonad with counit~$\eps$ and
comultiplication~$\delta$.
\end{lemma}
\begin{proof}
The comonad coassociativity law (right-hand diagram in~\eqref{dia:comonad_laws})
requires that the following diagram commutes:
    \begin{equation}
        \begin{tikzcd}
            X^\mon
                \ar[r, "\delta_X"]
                \ar[d, swap, "\delta_X"]
            &
            (X^\mon)^\mon
                \ar[d, "\delta_{X}^\mon"]
            \\
            (X^\mon)^\mon
                \ar[r, swap, "\delta_{X^\mon}"]
            &
            ((X^\mon)^\mon)^\mon
        \end{tikzcd}
    \end{equation}
    For an input $k \colon \mon \to X$, the upper path evaluates to:
    \begin{align*}
        (\delta_{X^\mon} \circ \delta_X)(k)
            & = \lam{u}{\lam{v}{\delta_X(k)(v \monop u)}}
            \aligncomment{Definition $\delta$
                     (see \eqref{eq:cowriter_comult}).}
            \\
            & = \lam{u}{\lam{v}{(\lam{m}{\lam{n}{k(n \monop m)}})(v \monop
            u)}}
            \aligncomment{Definition $\delta$
                     (see \eqref{eq:cowriter_comult}).}
            \\
            & = \lam{u}{\lam{v}{\lam{n}{k(n\monop v \monop u)}}}.
    \end{align*}
    This indeed agrees with the lower path:
    \begin{align*}
        (\delta_X^\mon \circ \delta_X)(k)
            & = \delta_X^\mon\left[
                \lam{u}{\lam{w}{k(w \monop u)}}\right]
            \aligncomment{Definition $\delta$
                     (see \eqref{eq:cowriter_comult}).}
            \\
            & = \lam{u}{\delta_X (\lam{w}{k(w \monop u)}})
            \aligncomment{Definition $\hole^\mon$ %
                     (see \eqref{eq:cowriter_on_morph}).}
            \\
            & = \lam{u}({\lam{m}{
                \lam{n}{(\lam{w}{k(w \monop u)})(n \monop m))} }}
            \aligncomment{Definition $\delta$
                     (see \eqref{eq:cowriter_comult}).}
            \\
            & = \lam{u}{\lam{m}{\lam{n}{k(n \monop m \monop u)}}},
            \\
            & = \lam{u}{\lam{v}{\lam{n}{k(n \monop v \monop u)}}},
            \aligncomment{Rename bound variables.}
    \end{align*}
    which shows that both paths in the diagram indeed agree.

    Finally, we need to show that the comonad counit law
    (left diagram in~\eqref{dia:comonad_laws}) holds:
    \begin{equation}
        \begin{tikzcd}[row sep = large]
            &
            X^\mon
                \ar[dl, swap, "\id"]
                \ar[d, "\delta_X"]
                \ar[dr, "\id"]
            \\
            X^\mon
            & (X^\mon)^\mon
                \ar[l, "\eps_{(X^\mon)}"]
                \ar[r,swap, "\eps_X^\mon"]
            & X^\mon
        \end{tikzcd}
    \end{equation}
    To this end, we show for both inner triangles separately that they commute
    on any input $k \colon \mon \to X$.
    The following computation shows that the right triangle commutes:
    \begin{align*}
        \eps_X^\mon(\delta_X(k))
            & = \eps_X^\mon (\lam{m}{\lam{n}{k(n \monop m)}})
            \aligncomment{Definition $\delta$ (see \eqref{eq:cowriter_comult}).}
            \\
            & = \lam{m}{\eps_X(\lam{n}{k(n \monop m)})}
            \aligncomment{Definition $\hole^\mon$
                     (see \eqref{eq:cowriter_on_morph}).}
            \\
            & = \lam{m}{k(\monid \monop m)}
            \aligncomment{Definition $\eps$ (see \eqref{eq:cowriter_counit}).}
            \\
            & = \lam{m}{k(m)}
            \\
            & = k.
    \end{align*}
    Commutativity of the left triangle follows by a similar
    computation,
    using the definitions of $\eps$ and $\delta$
    (\eqref{eq:cowriter_counit} and \eqref{eq:cowriter_comult}):
    \begin{align*}
        \eps_{(X^\mon)}(\delta_X(k))
        & = \eps_{(X^\mon)} (\lam{m}{\lam{n}{k(n \monop m)}})
        \\
        & = \lam{n}{k(n \monop \monid)}
        \\
        & = \lam{n}{k(n)}
        \\
        & = k.
    \end{align*}
  \end{proof}
}
%end takeout

\begin{proof}[\Cref{lemma:cowriter_is_left_action} on \cpageref{lemma:cowriter_is_left_action}]
An Eilenberg-Moore coalgebra~$\gamma\colon X\to\cowriterfunc X$ may be identified with its
transpose~$\lTransp{\gamma}\colon X\times M\to X$ defined by $\lTransp{\gamma}(x, m) = \gamma(x)(m)$.
Let
\begin{equation*}
\sigma_{A,B} \from A \times B \to B \times A
\end{equation*}
denote the symmetry isomorphism of the product.
We then define $\overline{\gamma} = \lTransp{\gamma} \comp \sigma_{M,X}$ and show that this is 
a left action of $M$ on $X$.
To this end, first note that
\begin{equation*}
  \overline{\gamma}(\monid, x) = \lTransp{\gamma}(x, \monid) = \gamma(x)(\monid) = \eps(\gamma(x))	= x,
\end{equation*}
using only the definition of~$\lTransp{\gamma}$,~$\eps$, and the left-hand diagram in~(\ref{dia:comonad_coalg_laws}).
It remains to be shown that for all~$n,m\in M$ and~$x\in X$ we have
\begin{equation}
\overline{\gamma}(m, \overline{\gamma}(n, x)) = \overline{\gamma}(m \monop n, x)
\tag{$*$}
\end{equation}
To this end, we compute as follows (eliding the indices on the symmetry isomorphism):
\begin{align*}
  \overline{\gamma}(m, \overline{\gamma}(n, x))
  & =\lTransp{\gamma}(\lTransp{\gamma}(x, n), m) \\
  &= \gamma(\gamma(x)(n))(m)
    \tag*{def. of~$\lTransp{\gamma}$ and $\sigma$}\\
  &= (\cowriterfunc\gamma\comp \gamma)(x)(n)(m)
    \tag*{def.~of~$\cowriterfunc\gamma$}\\
  &= (\delta \comp \gamma)(x)(n)(m)
    \tag*{right diagram in~(\ref{dia:comonad_coalg_laws})}\\
  &= \gamma(x)(m \monop n)
    \tag*{def.~of~$\delta$} \\
  &= \overline{\gamma}(m \monop n, x)
    \tag*{def. of~$\overline{\gamma}$}
\end{align*}
Thus, each Eilenberg-Moore coalgebra of the cowriter comonad induces a left action of~$M$ on its carrier set.

We next show that every homomorphism~$h\colon(X, \gamma)\to(Y, \xi)$ in~$\EM(\cowriterfunc)$ is an
equivariant map~$(X, \overline{\gamma})\to(Y, \overline{\xi})$. Indeed, we have
\begin{align*}
  h(\overline{\gamma}(m, x))
  & = h(\lTransp{\gamma}(x,m))
  = h(\gamma(x)(m))
  = (h \comp \gamma(x))(m)
  = (\cowriterfunc h \comp \gamma)(x)(m) \\
  & = (\xi \comp h)(x)(m)
  = \xi(h(x))(m)
  = \lTransp{\xi}(\sigma(m, h(x)))
  = \overline{\xi}(m, h(x))
\end{align*}
and thus $h$ is equivariant.
Thus, we can map $(X, \chi)$ to $(X, \overline{\chi})$ and $h$ to $h$, which is immediately functorial.
Finally, this mapping is also an isomorphism because taking the transpose and composing with the natural isomorphism $\sigma$ yields an isomorphism of categories.
\end{proof}

%%%%%%%%%%%%%%%%%%%%%%%%%%%%%%%%%%%%%%
%%%%%%%%%%%%%%%%%%%%%%%%%%%%%%%%%%%%%%
%%%%%%%%%%%%%%%%%%%%%%%%%%%%%%%%%%%%%%
%%%%%%%%%%%%%%%%%%%%%%%%%%%%%%%%%%%%%%
%%%%%%%%%%%%%%%%%%%%%%%%%%%%%%%%%%%%%%
%%%%%%%%%%%%%%%%%%%%%%%%%%%%%%%%%%%%%%
%%%%%%%%%%%%%%%%%%%%%%%%%%%%%%%%%%%%%%
%%%%%%%%%%%%%%%%%%%%%%%%%%%%%%%%%%%%%%
%Appendix on cellular automata
%
\section{Details for~\cref{sec:ca}}

\begin{proof}[\Cref{lem:orb-inv} on \cpageref{lem:orb-inv}]
We first show that if $a\colon M\to X$ is injective, then~$E^{a}\cong N$.
In particular, we will show that~$E^{a}$ is the diagonal~$\Delta(N) =\setDef{(n,m)\in N\times N}{n=m}$. 
It is clear that~$\Delta(N)\subseteq E^{a}$.
Moreover,~if~$(n,m)\in E^{a},$ then~$a\comp i(n) = a\comp i(m)$ hence also~$n = m$ since~$a\comp i\colon N\to X$ is injective being the composite of injective maps.  
We conclude that~$E^{a} = \Delta(N)\cong N$ if~$a$ is injective.

We will now show that~$I^a\cong[N,S]$.
To this end, it is enough to show that the diagram
\begin{equation}
\label{eq:lem:orb-inv}
  \begin{tikzcd}
    \intHom{N}{S} \rar{\id}
    & \intHom{N}{S}
    \ar[r, shift left, "c^{a}_{1}"{above}]
    \ar[r, shift right, "c^{a}_{2}"{below}]
    & \intHom{N\times N}{S}
  \end{tikzcd}
\end{equation}
is a equaliser diagram. 
To see that~(\ref{eq:lem:orb-inv}) commutes, we compute as follows for a map~$f\colon N\to S$ and an orbit~$(n,m)\in E^{a} = \Delta(N)$:
\begin{align}
c_1^a(f)(n,m) &= f(p_1^a(n, m))				\tag{by definition of~$c_1^a$}  \\
	              &= f(n)						\tag{by definition of~$p_1^a$} \\
	              &= f(m)					\tag{since~$n=m$} \\
	              &= f(p_2^a(n,m))				\tag{by definition of~$p_2^a$} \\
	              &= c_2^a(f)(n,m)				\tag{by definition of~$c_2^a$}
\end{align}
Finally,~$(\intHom{N}{S}, \id)$ has the universal property of an equaliser.
Indeed, given a set~$Z$ and a map~$z\colon Z\to [N, S]$ with~$c_1^a\comp z = c_2^a\comp z$, then~$z$ itself is the unique map satisfying the identity~$\id\comp z = z$.
We conclude that~(\ref{eq:lem:orb-inv}) is an equaliser diagram.
It follows that~$I^a\cong\intHom{N}{S}$, as desired.
\qed
\end{proof}

\begin{proof}[\Cref{lem:global-rule-natural} on \cpageref{lem:global-rule-natural}]
  Let $h \from \gamma \to \delta$ be a homomorphism and let us just write $h$ instead of $\U{h}$.
  Then we have for all $c \in \dual{\U{\delta}}$ that
  \begin{align*}
    \dual{h} (G_{\delta}(c))
    & = G_{\delta}(c) \comp h
      \tag*{by def. $\dual{h}$} \\
    & = e \comp C c \comp \delta \comp h
      \tag*{by def. $G_{\delta}$} \\
    & = e \comp C c \comp C h \comp \gamma
      \tag*{$h$ homomorphism} \\
    & = e \comp C (c \comp  h) \comp \gamma
      \tag*{$C$ functor} \\
    & = G_{\gamma}(c \comp h)
      \tag*{by def. $G_{\gamma}$} \\
    & = (G_{\gamma} \comp \dual{h})(c)
      \tag*{by def. $\dual{h}$}
  \end{align*}
  and thus $\dual{h} \comp G_{\delta} = G_{\gamma} \comp \dual{h}$.
  Hence, $G_{-}$ is natural.
\qed
\end{proof}
%%%%%%%%%%%%%%%%%%%%%%%%%%%%%%%%%%%%%%%%%
%%%%%%%%%%%%%%%%%%%%%%%%%%%%%%%%%%%%%%
%%%%%%%%%%%%%%%%%%%%%%%%%%%%%%%%%%%%%%
%%%%%%%%%%%%%%%%%%%%%%%%%%%%%%%%%%%%%%
%%%%%%%%%%%%%%%%%%%%%%%%%%%%%%%%%%%%%%
%%%%%%%%%%%%%%%%%%%%%%%%%%%%%%%%%%%%%%
%%%%%%%%%%%%%%%%%%%%%%%%%%%%%%%%%%%%%%
%%%%%%%%%%%%%%%%%%%%%%%%%%%%%%%%%%%%%%
%%%%%%%%%%%%%%%%%%%%%%%%%%%%%%%%%%%%%%
%Appendix on cellular simulation
%
\section{Details for~\cref{sec:sim}}

\begin{proof}[\Cref{lem:ca-category} on \cpageref{lem:ca-category}]
  The composition in $\CA$ is defined by
  \begin{equation*}
    (g, r) \comp (f, s) = (g \comp f, r \comp s) \, .
  \end{equation*}
  This is well-defined because $\CApre$ is a category and by functoriality of $\dual{(-)}$:
  \begin{equation*}
    \dual{(g \comp f)} \comp (r \comp s)
    = \dual{f} \comp \dual{g} \comp r \comp s
    = \dual{f} \comp s
    = \id
  \end{equation*}
  The identity morphism is given by $(\id_{\gamma}, \id_{\dual{\U{\gamma}}})$.

  We define $P$ as identity on objects and by $P(f, s) = f$ on morphisms.
  This is clearly a functor by definition.
  Given an isomorphism $f \from \gamma \to \delta$ with inverse $g$ in $\CApre$, we can obtain cellular morphisms $(f, \dual{g})$ and $(g, \dual{f})$ that are inverses of each other by functoriality of $\dual{(-)}$.
\end{proof}

\begin{proof}[\Cref{lem:global-rule-natural-full} on \cpageref{lem:global-rule-natural-full}]
  This follows immediately from \cref{lem:global-rule-natural} and $s$ being a section of $\dual{f}$:
  \begin{equation*}
    \dual{f} \comp G_{\delta} \comp s
    = G_{\gamma} \comp \dual{f} \comp s
    = G_{\gamma}
    \tag*{\qed}
  \end{equation*}
\end{proof}

\begin{proof}[\Cref{lem:behaviour-reflection} on \cpageref{lem:behaviour-reflection}]
  Let $c \in \dual{(\im f)}$.
  Then we have
  \begin{align*}
    (\dual{g} \comp G_{\gamma} \comp \dual{e_{f}})(c)
    & = \dual{g}(G_{\gamma}(c \comp e_{f}))
      \tag*{def. composition} \\
    & = \dual{g}(e \comp C(c \comp e_{f}) \comp \gamma)
      \tag*{def. $G_{\gamma}$} \\
    & = \dual{g}(e \comp Cc \comp \delta' \comp e_{f})
      \tag*{$e_{f}$ pre-morphism} \\
    & = \dual{g} (\dual{e_{f}} (e \comp Cc \comp \delta'))
      \tag*{def. $\dual{e_f}$} \\
    & = (\dual{g} \comp \dual{e_{f}}) (e \comp Cc \comp \delta')
      \tag*{def. composition} \\
    & = \dual{(e_{f} \comp g)} (e \comp Cc \comp \delta')
      \tag*{functoriality of $\dual{(-)}$} \\
    & = e \comp Cc \comp \delta'
      \tag*{$g$ section of $e_{f}$} \\
    & = G_{\delta'}(c)
  \end{align*}
  This shows that $\dual{g} \comp G_{\gamma} \comp \dual{e_{f}} = G_{\delta'}$, as desired.
  \qed
\end{proof}

\begin{proof}[\Cref{prop:pair-based-pre-morphisms-gives-iso} on \cpageref{prop:pair-based-pre-morphisms-gives-iso}]
  Suppose that $f \from (\gamma, x_{0}) \to (\delta, y_{0})$ and $g \from (\delta, y_{0}) \to (\gamma, x_{0})$ are based pre-cellular morphisms.
  That these are inverses only requires reasoning about the monoid action:
  Let $y \in \U{\delta}$.
  Because $\delta$ is reachable, there is $m \in M$ with $\delta_{1}(y_{0})(m) = y$.
  We then have
  \begin{align*}
    f(g(y))
    & = f(g(\delta_{1}(y_{0})(m))) \\
    & = f(\gamma_{1}(g(y_{0}))(m)) \\
    & = f(\gamma_{1}(x_{0})(m)) \\
    & = \delta_{1}(f(x_{0}))(m) \\
    & = \delta_{1}(y_{0})(m) \\
    & = y
  \end{align*}
  By a similar argument, we also have $g \comp f = \id$.
  Hence, $f$ and $g$ are inverses of each other.
  By \cref{lem:ca-category}, these lift to cellular isomorphisms.
  \qed
\end{proof}

\section{Details for~\cref{sec:relation-equivalence}}

% For an endofunctor~$F\colon\Set\to\Set$, weak preservation of pullbacks can be checked point-wise.
% Indeed, Gumm~\cite[Theorem 2.8]{Gumm01} explains that~$F$ weakly preserves pullbacks if and only if for every cospan
% $X\xleftarrow{f} Z\xra{g} Y$ the following condition is obtained: for all~$(u, v)\in FX\times FY$ such that
% $Ff(u) = Fg(v)$,
% there exists $w\in F(X\times_Z Y)\cong F\{(x, y) | f(x) = g(y)\}$ such that~$Fp_1(w) = u$ and~$Fp_2(w) = v$.
% We employ this characterization in the following.

\begin{proof}[\Cref{lem:weak-pb-preservation} on \cpageref{lem:weak-pb-preservation}]
  Suppose that we are given the following pullback in $\Set$.
  \begin{equation*}
    \begin{tikzcd}
      X \rar{p^{2}} \dar[swap]{p^{1}} & B \dar{g} \\
      A \rar{f} & D
    \end{tikzcd}
  \end{equation*}
  We will show that the following square is a weak pullback.
  \begin{equation*}
    \begin{tikzcd}
      CX \rar{Cp^{2}} \dar[swap]{Cp^{1}} & B \dar{Cg} \\
      CA \rar{Cf} & CD
    \end{tikzcd}
  \end{equation*}
  To this end, assume that we have a commuting square
  \begin{equation}
    \label{eq:weak-pb-assumption}
    \begin{tikzcd}
      Y \rar{q^{2}} \dar[swap]{q^{1}} & B \dar{Cg} \\
      CA \rar{Cf} & CD
    \end{tikzcd}
  \end{equation}
  and then we wish to show that there is a map $h \from Y \to CX$ with $Cp^{k} \comp h = q^{k}$.

  We proceed in two steps.
  First, we note that $\cowriterfunc$ is a right-adjoint functor and thus preserves all limits.
  This means that we obtain a unique $h_{1} \from Y \to \cowriterfunc X$ with $\cowriterfunc p^{k} \comp h_{1} = q^{k}_{1}$.

  Second, let $y \in Y$ and $f \in I^{h_{1}(y)}$.
  The key is to note that
  \begin{equation*}
    I^{h_{1}(y)} = I^{q^{1}_{1}(y)} \cap I^{q^{2}_{1}(y)} \, ,
  \end{equation*}
  which allows us to apply $q^{k}_{2}(y)$ to $f$.
  By commutativity of \cref{eq:weak-pb-assumption}, we have $q^{1}_{2}(y)(f) = q^{2}_{2}(y)(f)$ and we can thus define
  \begin{equation*}
    h_{2}(y)(f) = q^{1}_{2}(y)(f) \, .
  \end{equation*}
  We can then put $h(y) = (h_{1}(y), h_{2}(y))$ and get $Cp^{k} \comp h = q^{k}$ by definition.
  \qed
\end{proof}

\begin{proof}[\Cref{prop:consistent-config-final} on \cpageref{prop:consistent-config-final}]
  Suppose that $(\rho, q, R, p, g)$ is a cellular bisimulation, which means in particular that $\dual{q_{1}} \comp p_{1} = \dual{q_{2}} \comp p_{2}$.
  By the construction of $\Cons$ as the pullback in \cref{eq:consistent-config-pullback}, we thus get a unique map $h \from R \to \Cons$ with $\ConsProj_{k} \comp h = p_{k}$.
  In order to show $h \comp g = \ConsCoalg \comp h$, we first note that
  \begin{align*}
    \pair{\ConsProj_{1}}{\ConsProj_{2}} \comp h \comp g
    & = \pair{\ConsProj_{1} \comp h}{\ConsProj_{2} \comp h} \comp g
      \tag*{by property pairing} \\
    & = \pair{p_{1}}{p_{2}} \comp g
      \tag*{by def. $h$} \\
    & = (G_{\gamma} \times G_{\delta}) \comp \pair{p_{1}}{p_{2}}
      \tag*{\cref{eq:cellular-bisimulation-diagrams} for $R$ and $g$} \\
    & = (G_{\gamma} \times G_{\delta}) \comp \pair{\ConsProj_{1}}{\ConsProj_{2}} \comp h
      \tag*{by def. $h$} \\
    & = \pair{\ConsProj_{1}}{\ConsProj_{2}} \comp \ConsCoalg \comp h
      \tag*{\cref{eq:cellular-bisimulation-diagrams} for $\mathrm{Cons}_{\gamma,q}$ and $\ConsCoalg$}
  \end{align*}
  Since the pairing $\pair{\ConsProj_{1}}{\ConsProj_{2}} \from \Cons \to \dual{\U{\delta}} \times \dual{\U{\gamma}}$ is a monomorphism by the pullback mapping property, we thus obtain $h \comp g = \ConsCoalg \comp h$.
  \qed
\end{proof}

%%%%%%%%%%%%%%%%%%%%%%%%%%%%%%%%%%%%%%
%%%%%%%%%%%%%%%%%%%%%%%%%%%%%%%%%%%%%%
%%%%%%%%%%%%%%%%%%%%%%%%%%%%%%%%%%%%%%
%%%%%%%%%%%%%%%%%%%%%%%%%%%%%%%%%%%%%%
%%%%%%%%%%%%%%%%%%%%%%%%%%%%%%%%%%%%%%
%%%%%%%%%%%%%%%%%%%%%%%%%%%%%%%%%%%%%%
%%%%%%%%%%%%%%%%%%%%%%%%%%%%%%%%%%%%%%
%%%%%%%%%%%%%%%%%%%%%%%%%%%%%%%%%%%%%%
%Appendix on logic
%
\section{Details for~\cref{sec:logic}}

\begin{proof}[\Cref{lemma:space_diamonds_and_monoid_mult} on \cpageref{lemma:space_diamonds_and_monoid_mult}]
  Given a based CA $(\gamma, x)$ and a configuration $c$, using the induced monoid action $\overline{\gamma}$ with $\overline{\gamma}(m, x) = \gamma_{1}(x)(m)$ from \cref{lemma:cowriter_is_left_action}, we have the following.
  \begin{align*}
    (\gamma, x), c \sems\spacedia{m}\spacedia{n}\varphi
    \quad&\text{iff}\quad (\gamma, \overline{\gamma}(m, x), c) \sems\spacedia{n}\varphi \\
    \quad&\text{iff}\quad (\gamma, \overline{\gamma}(n, \overline{\gamma}(m, x)), c) \sems\varphi \\
    \quad&\text{iff}\quad (\gamma, \overline{\gamma}(m \monop n, x), c) \sems\varphi \\
    \quad&\text{iff}\quad (\gamma, x, c) \sems\spacedia{m\monop n}\varphi
  \end{align*}
  Since this holds for arbitrary cellular automata and configurations, we conclude that $\spacedia{m} \spacedia{n} \varphi\equiv\spacedia{m \monop n} \varphi$, as desired.
  \qed
\end{proof}

\begin{proof}[\Cref{lem:logical-transport-category} on \cpageref{lem:logical-transport-category}]
Obviously, we have that $(\gamma, x_{0}, c) \sems \phi$ holds if and only if $(\gamma, x_{0}, \id(c)) \sems \phi$ holds.
  Given logical transports $t_{1}\from (\gamma, x_{0}) \logTrans (\delta, y_{0})$ and  $t_{2}\from (\delta, y_{0}) \logTrans (\rho, z_{0})$,
  we obtain
    \begin{equation*}
    (\gamma, x_{0}, c) \sems \phi
    \quad \text{iff} \quad
    (\delta, y_{0}, t_{1}(c)) \sems \phi
    \quad \text{iff} \quad
    (\rho, z_{0}, t_{2}(t_{1}(c))) \sems \phi
  \end{equation*}
  and thus the composition $t_{2} \comp t_{1} \from \dual{\U{\gamma}} \to \dual{\U{\rho}}$ is a logical transport.
  \qed
\end{proof}

\begin{proof}[\Cref{thm:cellular-morphism-gives-logical-transport} on \cpageref{thm:cellular-morphism-gives-logical-transport}]
  Let $(f,s) \from (\gamma, x_{0}) \to (\delta, y_{0})$ be a based cellular morphism.
  \Cref{ex:cellular-bisimulation-from-morphism} yields a cellular bisimulation $(\rho_{f}, q, \mathrm{Gr}_{\dual{f}}, p, g) \from \proMor{\gamma}{\delta}$, and from \cref{thm:bisimulation-to-logical-equivalence} we get
  \begin{equation}
    \label{eq:graph-logical-equivalence}
    (\gamma, q_{1}(u), p_{1}(c)) \sems \phi
    \iff
    (\delta, q_{2}(u), p_{2}(c)) \sems \phi
  \end{equation}
  for all $u \in \mathrm{Gr}_{f}$ and $c \in R$.
  Since $\mathrm{Gr}_{f}$ is the graph of $f$, for every $x \in X$ there is a unique $u \in \mathrm{Gr}_{f}$ with $q_{1}(u) = x$ and $q_{2}(u) = f(x)$.
  Finally, we constructed in \cref{ex:cellular-bisimulation-from-morphism} a map $t_{1} \from \dual{\U{\gamma}} \to \mathrm{Gr}_{\dual{f}}$ with $p_{1} \comp t_{1} = \id$ and $p_{2} \comp t_{1} = s$.
  Therefore, we obtain from \cref{eq:graph-logical-equivalence} for all $c \in \dual{\U{\gamma}}$ that
  \begin{equation*}
    (\gamma, x_{0}, c) \sems \phi
    \iff
    (\delta, f(x_{0}), s(c)) \sems \phi
  \end{equation*}
  and thus $s$ is a logical transport $(\gamma, x_{0}) \logTrans (\delta, y_{0})$ as required.
  \qed
\end{proof}

\begin{proof}[\cref{thm:bisimulation-to-logical-equivalence} on \cpageref{thm:bisimulation-to-logical-equivalence}]
  Let $(\rho, f, R, p, g)$ be a cellular bisimulation $\proMor{\gamma}{\delta}$.
  We proceed by induction on formulas that for all $u \in \U{\rho}$, $c \in R$ and formulas $\phi$, we have
  \begin{equation*}
    (\gamma, f_{1}(u), p_{1}(c)) \sems \phi
    \iff
    (\delta, f_{2}(u), p_{2}(c)) \sems \phi
  \end{equation*}
  This gives then immediately that $(R, p)$ is a logical equivalence between $(\gamma, f_{1}(u))$ and $(\delta, f_{2}(u))$ for all $u \in \U{\rho}$.
  \begin{itemize}
  \item Base case:
    First, we note that by being a cellular bisimulation, we have that
    \begin{equation*}
      p_{1}(c) \comp f_{1}
      = (\dual{f_{1}} \comp p_{1})(c)
      = (\dual{f_{2}} \comp p_{2})(c)
      = p_{2}(c) \comp f_{2} \, .
    \end{equation*}
    Thus, for all $s \in S$, by definition of the semantics,
    \begin{align*}
      (\gamma, f_{1}(u), p_{1}(c)) \sems s
      & \iff p_{1}(c)(f_{1}(u)) = s \\
      & \iff p_{2}(c)(f_{2}(u)) = s \\
      & \iff (\gamma, f_{2}(u), p_{2}(c)) \sems s
    \end{align*}
  \item The propositional connectives are immediate by definition and using the induction hypothesis.
  \item Spatial modality: For $m \in M$, we have
    \begin{align*}
      (\gamma, f_{1}(u), p_{1}(c)) \sems \spacedia{m} \varphi
      & \iff (\gamma, \gamma_{1}(f_{1}(u))(m), p_{1}(c)) \sems \varphi
        \tag*{by definition} \\
      & \iff (\gamma, f_{1}(\rho_{1}(u)(m)), p_{1}(c)) \sems \varphi
        \tag*{by equivariance} \\
      & \iff (\delta, f_{2}(\rho_{1}(u)(m)), p_{2}(c)) \sems \varphi
        \tag*{by IH} \\
      & \iff (\delta, \delta_{1}(f_{2}(u))(m), p_{2}(c)) \sems \varphi
        \tag*{by equivariance} \\
      & \iff (\gamma, f_{2}(u), p_{2}(c)) \sems \spacedia{m} \varphi
        \tag*{by definition}
    \end{align*}
  \item Update modality:
    By definition of cellular bisimulations, we get $g(c) \in R$ with
    \begin{equation}
      \label{eq:configuration-step}
      p(g(c)) = (G_{\gamma}(p_{1}(c)), G_{\delta}(p_{2}(c))) \, ,
    \end{equation}
    and thus
    \begin{align*}
      (\gamma, f_{1}(u), p_{1}(c)) \sems \timedia \varphi
      & \iff (\gamma, f_{1}(u), G_{\gamma}(p_{1}(c))) \sems \varphi
        \tag*{by definition} \\
      & \iff (\gamma, f_{1}(u), p_{1}(g(c))) \sems \varphi
        \tag*{\cref{eq:configuration-step}} \\
      & \iff (\delta, f_{2}(u), p_{2}(g(c))) \sems \varphi
        \tag*{by IH} \\
      & \iff (\delta, f_{2}(u), G_{\delta}(p_{2}(c))) \sems \varphi
        \tag*{\cref{eq:configuration-step}} \\
      & \iff (\delta, f_{2}(u), p_{2}(c)) \sems \timedia \varphi
        \tag*{by definition}
    \end{align*}
  \end{itemize}
  This concludes the induction and thus the proof.
  \qed
\end{proof}

%%%%%%%%%%%%%%%%%%%%%%%%%%%%%%%%%%%%%% 
\section{Details for~\cref{sec:hennessy_milner}}

\begin{proof}[\Cref{lem:lifting-local-configurations} on \cpageref{lem:lifting-local-configurations}]
  Let $a \in S$ be the inhabitant and set
  \begin{equation*}
    h(x) =
    \begin{cases}
      f(n), & x = g(n) \\
      a, \text{otherwise}
    \end{cases} \, .
  \end{equation*}
  This is well-defined by the assumption that on $f$ and clearly $h(g(n)) = f(n)$. \qed
\end{proof}

\begin{proof}[\Cref{theo:hm_theorem} on \cpageref{theo:hm_theorem}]
  Given a logical transport $t \from (\gamma, x_0) \logTrans (\delta, y_0)$, we will construct a based cellular morphism $\bar{t} \from (\gamma, x_0) \to (\delta, y_0)$ with $\pi_{2}(\bar{t}) = t$.
  This means that we have to show that there exists a (necessarily unique) based pre-cellular morphism $f_t \from (\gamma, x_0) \to (\delta, y_0)$, such that $t$ is a section of $\dual{f_t}$.
  For brevity, let us write $X = \U{\gamma}$ and $Y = \U{\delta}$.

  First, note that since $(\gamma, x_0)$ is based and reachable, every cell $x \in X$ has the form $\gamma_1(x_0)(m)$ for some $m\in M$.
  Since $\gamma_{1}$ is free, this $m$ is unique and we can define $f_{t}$ as the unique map with the following graph.
  \begin{equation*}
    \setDef{(x, \delta_{1}(y_{0})(m)) \in X \times Y}{m \in M, x = \gamma_{1}(x_{0})(m)}
  \end{equation*}
  Moreover, $f_t$ is equivariant by definition.

  Let us next show that $t$ is a section for $\dual{f_{t}}$.
  Given a configuration $c$, we need to show that $t(c \comp f_{t})(y) = c(y)$ for all $y \in Y$.
  Since $\delta$ is reachable, there is $m \in M$ with $\delta_{1}(y_{0})(m) = y$.
  Then we have
  \begin{align*}
    t(c \comp f_{t})(y) = s
    & \iff t(c \comp f_{t})(\delta_{1}(y_{0})(m)) = s
      \tag*{def. $m$} \\
    & \iff (\delta, \delta_{1}(y_{0})(m), t(c \comp f_{t})) \sems s
      \tag*{def. semantics} \\
    & \iff (\delta, y_{0}, t(c \comp f_{t})) \sems \spacedia{m} s
      \tag*{def. semantics} \\
    & \iff (\gamma, x_{0}, c \comp f_{t}) \sems \spacedia{m} s
      \tag*{logical transport} \\
    & \iff (\gamma, \gamma_{1}(x_{0})(m), c \comp f_{t}) \sems s
      \tag*{def. semantics} \\
    & \iff (c \comp f_{t})(\gamma_{1}(x_{0})(m)) = s
      \tag*{def. semantics}
  \end{align*}
  From the last line, we obtain
  \begin{equation*}
    s
%    = (c \comp f_{t})(\gamma_{1}(x_{0})(m))
    = c (f_{t}(\gamma_{1}(x_{0})(m)))
    = c (\delta_{1}(f_{t}(x_{0}))(m))
    = c (\delta_{1}(y_{0})(m))
    = c(y)
  \end{equation*}
  and thus $t(c \comp f_{t}) = c$.

  We even have that $\dual{f_{t}}$ and $t$ are inverses.
  Since $\delta$ is reachable, every $y \in Y$ is of the form $\delta_{1}(y_{0})(m)$.
  Therefore, we have
  \begin{equation*}
    f_{t}(\gamma_{1}(x_{0})(m)) = \delta_{1}(y_{0})(m) = y
  \end{equation*}
  and thus $f_{t}$ is an epimorphism (surjective).
  For all $c, d \in \dual{Y}$ with $\dual{f_{t}}(c) = \dual{f_{t}}(d)$, we obtain $c = d$ because $f_{t}$ is epi.
  Hence, $\dual{f_{t}}$ is a monomorphism (injective), which implies that it is an isomorphism with $t$ as inverse.
  % Using that $t$ is a section for $\dual{f_{t}}$, we have
  % \begin{equation*}
  %   \dual{f_{t}} \comp t \comp \dual{f_{t}} = \dual{f_{t}} = \dual{f_{t}} \comp \id
  % \end{equation*}
  % and thus $t \comp \dual{f_{t}} = \id$, since $\dual{f_{t}}$ is mono.

  Next, we construct the graph of $t$ as the following pullback.
  \begin{equation*}
    \begin{tikzcd}
      \mathrm{Gr}_{t} \rar{p_{2}} \dar{p_{1}}
      & \dual{Y} \dar{\id} \\
      \dual{X} \rar{t}
      & \dual{Y}
    \end{tikzcd}
  \end{equation*}
  Since the following two diagrams commute by the identity law and $t$ and $\dual{f_{t}}$ being inverses, we get sections $s_{1}$ and $s_{2}$ of the projections $p_{1}$ and $p_{2}$ from the mapping property of the pullback.
  \begin{equation*}
    \begin{tikzcd}
      \dual{X} \dar{\id} \rar{t}
      & \dual{Y} \dar{\id} \\
      \dual{X} \rar{t}
      & \dual{Y}
    \end{tikzcd}
    \qquad
    \begin{tikzcd}
      \dual{Y} \rar{\id} \dar{\dual{f_{t}}}
      & \dual{Y} \dar{\id} \\
      \dual{X} \rar{t}
      & \dual{Y}
    \end{tikzcd}
  \end{equation*}
  This makes $(\mathrm{Gr}_{t}, p, s)$ a logical equivalence, using the assumption that $t$ is a logical transport.
  \Cref{thm:relational-hm} gives us a cellular bisimulation $(\rho, q, \Cons, p', g)$ with $(x_{0}, y_{0}) \in \U{\rho}$.

  Finally, we define $r \from X \to \U{\rho}$ as follows, again appealing to $\gamma_{1}$ being free and to reachability.
  \begin{equation*}
    r(x) = \rho_{1}(x_{0}, y_{0})(m), \quad \text{for } x = \gamma_{1}(x)(m)
  \end{equation*}
  Since $q_{2} \comp r = f_{t}$ by equivariance of $q_{2}$, it remains to show that $r$ is pre-cellular to obtain that $f_{t}$ is a pre-cellular morphism.
  For $x = \gamma_{1}(x)(m)$ we have
  \begin{align*}
    r(\gamma_{1}(x)(m'))
    & = r(\gamma_{1}(\gamma_{1}(x_{0}(m)))(m'))
      \tag*{definition of $x$} \\
    & = r(\gamma_{1}(x_{0})(m' \monop m))
      \tag*{monoid action} \\
    & = \rho_{1}(x_{0}, y_{0})(m' \monop m)
      \tag*{definition of $r$} \\
    & = \rho_{1}(\rho_{1}(x_{0}, y_{0})(m))(m')
      \tag*{monoid action} \\
    & = \rho_{1}(r(x))(m')
      \tag*{definition of $r$}
  \end{align*}
  and
  \begin{align*}
    \rho_{2}(r(x))(f)
    & = \rho_{2}(\gamma_{1}(x_{0}, y_{0})(m))(f)
      \tag*{definition of $r$} \\
    & = \delta_{2}(q_{1}(\gamma_{1}(x_{0}, y_{0})(m)))(f)
      \tag*{$q_{1}$ pre-cellular} \\
    & = \delta_{2}(\delta_{1}(q_{1}(x_{0}, y_{0}))(m))(f)
      \tag*{$q_{1}$ equivariant} \\
    & = \delta_{2}(\delta_{1}(x_{0})(m))(f)
      \tag*{$q_{1}$ projection} \\
    & = \delta_{2}(x)(f)
      \tag*{definition of $x$}
  \end{align*}
  This concludes the proof that $f_{t}$ is a pre-cellular morphism and thus the pair $(f_{t}, t)$ is a cellular morphism.
  \qed
\end{proof}

\begin{proof}[\Cref{thm:relational-hm} on \cpageref{thm:relational-hm}]
  Let $X = \U{\gamma}$ and $Y = \U{\delta}$.
  We define the common orbit of cells $x \in X$ and $y \in Y$ in $\gamma$ and $\delta$ to be the image of the product action $\pair{\gamma_{1}(x)}{\delta_{1}(y)}$, given by
  \begin{equation*}
    O_{(x,y)} = \setDef{(\gamma_{1}(x)(m), \delta_{1}(y)(m))}{m \in M}.
  \end{equation*}
  Define a relation $Q \subseteq X \times Y$ by
  \begin{equation*}
    Q = \setDef{(x, y) \in O_{(x_{0}, y_{0})}}{
      \all{\phi} \all{c \in R} (\gamma, x, p_{1}(c)) \sems \phi \iff (\delta, y, p_{2}(c)) \sems \phi}
  \end{equation*}
  together with projection maps $q_{1} \from Q \to X$ and $q_{2} \from Q \to Y$.
  Clearly, $(x_{0}, y_{0}) \in Q$ because they are in their own orbit and by $R$ being a logical equivalence.
  We prove that $Q$ carries a cellular bisimulation and that the configurations in $R$ are consistent in three steps.

  \paragraph*{Step 1 -- Coalgebra on $Q$.}
  We prove that there is a coalgebra $\rho \from Q \to CQ$.
  In order to define $\rho_{1}$, let $\Delta \from \cowriterfunc X \times \cowriterfunc Y \to \cowriterfunc (X \times Y)$ be the monoidal strength with $\Delta(a, b) = \pair{a}{b}$ and set
  \begin{equation*}
    \rho_{1} = X \times Y \xrightarrow{\gamma_{1} \times \delta_{1}} \cowriterfunc X \times \cowriterfunc Y \xrightarrow{\Delta} \cowriterfunc (X \times Y) \, .
  \end{equation*}
  With this definition, $\rho$ is a monoid action.
  It remains to show that it restricts to an action on $Q$, that is, that for all $(x, y) \in Q$ we will show that $\rho_{1}(x, y)(m) \in R$.
  Since $(x, y) \in O_{(x_{0}, y_{0})}$ and because $\gamma_{1}$ and $\delta_{1}$ are monoid actions, we have
  \begin{align*}
    \rho_{1}(x,y)(m)
    & = \rho_{1}(x,y)(\pair{\gamma_{1}(x_{0})}{\delta_{1}(y_{0})}(m'))(m)
      \tag*{$(x, y) \in O_{(x_{0}, y_{0})}$} \\
    & = (\gamma_{1} \times \delta_{1})(\pair{\gamma_{1}(x_{0})}{\delta_{1}(y_{0})}(m'))(m)
      \tag*{by def.} \\
    & = \pair{\gamma_{1} \comp \gamma_{1}(x_{0})}{\delta_{1} \comp \delta_{1}(y_{0})}(m')(m)
      \tag*{property pairing} \\
    & = \pair{\gamma_{1}(x_{0})}{\delta_{1}(y_{0})}(m \monop m')
      \tag*{monoid action}
  \end{align*}
  and thus $\rho_{1}(x,y)(m) \in O_{(x_{0}, y_{0})}$.
  For any formula $\phi$ and $c \in R$, we have
  \begin{align*}
    (\gamma, \gamma_{1}(x)(m), p_{1}(c)) \sems \phi
    & \iff (\gamma, x, p_{1}(c)) \sems \spacedia{m} \phi
      \tag*{def. semantics} \\
    & \iff (\delta, y, p_{2}(c)) \sems \spacedia{m} \phi
      \tag*{$(x, y) \in Q$} \\
    & \iff (\delta, \delta_{1}(x)(m), p_{1}(c)) \sems \phi
      \tag*{def. semantics}
  \end{align*}
  and thus $\rho_{1}(x,y)(m) \in R$ for all $m \in M$.
  This means that $\rho_{1}$ is a well-defined map $Q \to \cowriterfunc Q$.

  To define $\rho_{2}(x, y)$, we need that $I^{\rho_{1}(x,y)} \subseteq I^{\gamma_{1}(x)}$ for all $(x, y) \in O_{(x_{0}, y_{0})}$.
  Given $f \in I^{\rho_{1}(x,y)}$, we show that $f \in I^{\gamma_{1}(x)}$.
  If $|S| \leq 1$, then this is immediately clear because then $I^{\gamma_{1}(x)} = \intHom{N}{S}$ as there is at most one map into $S$.
  Thus, we assume that there are $a, b \in S$ with $a \neq b$.
  Suppose that $f \in I^{\rho_{1}(x,y)}$, $n,n' \in N$ and $\gamma_{1}(x)(n) = \gamma_{1}(x)(n')$.
  We need to show that $f(n) = f(n')$.
  First, we define a configuration $c_{2} \in \dual{Y}$ by $c_{2}(\delta_{1}(y)(n')) = a$ and $c_{2}(z) = b$ for all other $z \in Y$.
  Then we set $c = s_{2}(c_{2})$ and $c_{1} = p_{1}(c)$.
  Since $(x, y) \in O_{(x_{0}, y_{0})}$, there is some $m \in M$ with $x = \gamma_{1}(x_{0})(m)$ and $y = \delta_{1}(y_{0})(m)$.
  We then have
  \begin{align*}
    & c_{1}(\gamma_{1}(x)(n)) = c_{1}(\gamma_{1}(x)(n'))
    \tag*{since $\gamma_{1}(x)(n) = \gamma_{1}(x)(n')$} \\
    & \iff
      (\gamma, x, c_{1}) \sems \spacedia{n} c_{1}(\gamma_{1}(x)(n'))
      \tag*{def. semantics} \\
    & \iff
      (\gamma, x_{0}, c_{1}) \sems \spacedia{m} \spacedia{n} c_{1}(\gamma_{1}(x)(n'))
      \tag*{def. semantics} \\
    & \iff
      (\delta, y_{0}, c_{2}) \sems \spacedia{m} \spacedia{n} c_{1}(\gamma_{1}(x)(n'))
      \tag*{logical equivalence} \\
    & \iff
      (\delta, y, c_{2}) \sems \spacedia{n} c_{1}(\gamma_{1}(x)(n'))
      \tag*{def. semantics} \\
    & \iff c_{2}(\delta_{1}(y)(n)) = c_{1}(\gamma_{1}(x)(n'))
      \tag*{def. semantics}
  \end{align*}
  and similarly
  \begin{align*}
    & c_{1}(\gamma_{1}(x)(n')) = c_{1}(\gamma_{1}(x)(n'))
    \tag*{reflexivity of equality} \\
    & \iff
      (\gamma, x, c_{1}) \sems \spacedia{n'} c_{1}(\gamma_{1}(x)(n'))
      \tag*{def. semantics} \\
    & \iff
      (\gamma, x_{0}, c_{1}) \sems \spacedia{m} \spacedia{n'} c_{1}(\gamma_{1}(x)(n'))
      \tag*{def. semantics} \\
    & \iff
      (\delta, y_{0}, c_{2}) \sems \spacedia{m} \spacedia{n'} c_{1}(\gamma_{1}(x)(n'))
      \tag*{logical equivalence} \\
    & \iff
      (\delta, y, c_{2}) \sems \spacedia{n'} c_{1}(\gamma_{1}(x)(n'))
      \tag*{def. semantics} \\
    & \iff c_{2}(\delta_{1}(y)(n')) = c_{1}(\gamma_{1}(x)(n'))
      \tag*{def. semantics}
  \end{align*}
  Thus, we obtain
  \begin{equation*}
    c_{2}(\delta_{1}(y)(n)) = c_{1}(\gamma_{1}(x)(n')) = c_{2}(\delta_{1}(y)(n'))
  \end{equation*}
  and so, by definition of $c_{2}$, we must have that $\delta_{1}(y)(n) = \delta_{1}(y)(n')$.
  Using the definition of $\rho$, we therefore obtain
  \begin{equation*}
    \rho(x, y)(n) = (\gamma_{1}(x)(n), \delta_{1}(y)(n)) = (\gamma_{1}(x)(n'), \delta_{1}(y)(n')) = \rho(x, y)(n')
  \end{equation*}
  and, since $f \in I^{\rho(x, y)}$, thus $f(n) = f(n')$.
  Since this holds for all $n, n'$ with $\gamma_{1}(x)(n) = \gamma_{1}(x)(n')$, we obtain $f \in I^{\gamma_{1}(x)}$.

  % we first note that
  % \begin{align*}
  %   (n, n') \in E^{\rho_{1}(x,y)}
  %   & \iff \rho_{1}(x,y)(n) = \rho_{1}(x,y)(n') \\
  %   & \iff \gamma_{1}(x)(n) = \gamma_{1}(x)(n') \text{ and } \delta_{1}(x)(n) = \delta_{1}(x)(n') \\
  %   & \iff (n, n') \in E^{\gamma_{1}(x)} \cap E^{\delta_{1}(y)}
  % \end{align*}
  Since $f \in I^{\rho_{1}(x,y)}$ implies $f \in I^{\gamma_{1}(x)}$, we can set
  \begin{equation*}
    \rho_{2}(x, y)(f) = \gamma_{2}(x)(f) \, .
  \end{equation*}
  Thus, we obtain the sought cellular automaton on $Q$ by defining $\rho(x,y) = (\rho_{1}(x,y), \rho_{2}(x,y))$.

  \paragraph*{Step 2 -- Projections are pre-cellular morphisms.}
  We show that $q_{1} \from Q \to X$ and $q_{2} \from Q \to Y$ are pre-cellular morphisms.
  It is clear that both are equivariant by the definition of $\rho_{1}$.
  Moreover, we have
  \begin{equation*}
    \rho_{2}(x, y)(t_{q_{1}}^{\rho(x,y)} (f)) = \gamma_{2}(x)(f)
  \end{equation*}
  by definition of $\rho_{2}$, and thus $Cq_{2} \comp \rho = \gamma \comp q_{2}$.

  In order to prove that $q_{2}$ is also a pre-cellular morphism it remains to show that $\rho_{2}(x, y)(t_{q_{2}}^{\rho(x,y)} (f)) = \delta_{2}(y)(f)$.
  Since $(x, y) \in O_{(x_{0}, y_{0})}$, we have an $m \in M$ with $(x, y) = \pair{\gamma_{1}(x_{0})}{\delta_{1}(y_{0})}(m)$.
  By \cref{lem:lifting-local-configurations}, there is a $c \in \dual{Y}$ with $c \comp \delta_{1}(y) \comp i = f$.
  Using that the logical equivalence has a section $s_{2} \from \dual{Y} \to R$ for $p_{2}$, we obtain $d = s_{2}(c)$ with $p_{2}(d) = c$.
  For this configuration, we have
  \begin{align*}
    & c \comp \delta_{1}(y) \comp i = f \\
    & \iff (\delta, y, c) \sems \infconj_{n \in N} \spacedia{n} f(n)
      \tag*{by \cref{lem:configuration-test-formula}} \\
    & \iff (\gamma, y_{0}, c) \sems \spacedia{m} \parens*{\infconj_{n \in N} \spacedia{n} f(n)}
      \tag*{$y = \delta_{1}(y_{0})(m)$ and def. semantics} \\
    & \iff (\delta, x_{0}, p_{1}(d)) \sems \spacedia{m} \parens*{\infconj_{n \in N} \spacedia{n} f(n)}
      \tag*{$c = p_{2}(d)$ and logical equivalence} \\
    & \iff (\delta, x, p_{1}(d)) \sems \infconj_{n \in N} \spacedia{n} f(n)
      \tag*{$x = \gamma_{1}(x_{0})(m)$ and def. semantics} \\
    & \iff p_{1}(d) \comp \gamma_{1}(x) \comp i = f
      \tag*{by \cref{lem:configuration-test-formula}}
  \end{align*}
  and thus $p_{1}(d)$ also extends $f$.
  This allows us to finalise the pre-morphism proof.
  For $s \in S$, we define $\phi = \parens*{\infconj_{n \in N} \spacedia{n} f(n)} \to \timedia s$ and then we have
  \begin{align*}
    \delta_{2}(y)(f) = s
    & \iff (\delta, y, c) \sems \phi
      \tag*{by \cref{lem:local-rule-test-formula}} \\
    & \iff (\delta, \delta_{1}(y_{0})(m), c) \sems \phi
      \tag*{since $\delta_{1}(y_{0})(m) = y$} \\
    & \iff (\delta, y_{0}, c) \sems \spacedia{m} \phi
      \tag*{def. semantics} \\
    & \iff (\gamma, x_{0}, p_{1}(d)) \sems \spacedia{m} \phi
      \tag*{logical equivalence} \\
    & \iff (\gamma, x, p_{1}(d)) \sems \phi
      \tag*{$\gamma_{1}(x_{0})(m) = x$ and semantics} \\
    & \iff \gamma_{2}(x)(f) = s
      \tag*{$p_{1}(d) \comp \gamma_{1}(x) \comp i = f$ and \cref{lem:local-rule-test-formula}}
  \end{align*}
  This and the definition of $\rho$ yields
  \begin{equation*}
    \rho_{2}(x, y)(f) = \gamma_{2}(x)(f) = \delta_{2}(y)(f)
  \end{equation*}
  and thus $q_{2}$ is also a pre-cellular morphism and $(\rho, q)$ is a pre-cellular bisimulation $\proMor{\gamma}{\delta}$.
  By \cref{ex:consistent-configuration-bisimulation}, $(\rho, q, \mathrm{Cons}_{\rho,q}, p', g)$ is a cellular bisimulation.

  \paragraph*{Step 3 -- Map into consistent configurations.}
  Lastly, we show that there is a map $h \from R \to \mathrm{Cons}_{\rho,q}$.
  We do this by showing that $\dual{q_{1}} \comp p_{1} = \dual{q_{2}} \comp p_{2}$ and then applying the mapping property of the pullback defining $\mathrm{Cons}_{\rho,q}$.
  Let $c \in R$ and set $c_{k} = p_{k}(c)$.
  For $(x, y) \in Q$ we have by definition that $(\dual{q_{1}} \comp p_{1})(c)(x, y) = c_{1}(x)$ and $(\dual{q_{2}} \comp p_{2})(c)(x, y) = c_{2}(y)$.
  Then for $s \in S$, have
  \begin{align*}
    c_{1}(x) = s
    & \iff (\gamma, x, c_{1}) \sems s
      \tag*{def. semantics} \\
    & \iff (\gamma, x_{0}, c_{1}) \sems \spacedia{m} s
      \tag*{$x = \gamma_{1}(x_{0})(m)$ and semantics} \\
    & \iff (\delta, y_{0}, c_{2}) \sems \spacedia{m} s
      \tag*{logical equivalence} \\
    & \iff (\delta, y, c_{2}) \sems s
      \tag*{$y = \delta_{1}(y_{0})(m)$ and semantics} \\
    & \iff c_{2}(y) = s
      \tag*{def. semantics}
  \end{align*}
  This shows that $\dual{q_{1}} \comp p_{1} = \dual{q_{2}} \comp p_{2}$ and we obtain a unique $h \from R \to \mathrm{Cons}_{\rho,q}$ with $p_{k}' \comp h = p_{k}$.
  \qed
\end{proof}

%%% Local Variables:
%%% mode: LaTeX
%%% TeX-master: "main-esop-revised"
%%% End:

\end{document}